\documentclass[11pt]{article}
\usepackage[margin=1in]{geometry}
\usepackage{amssymb}
\usepackage{amsmath}
\usepackage{amstext}
\usepackage{amsthm}
\usepackage{bbold}
\usepackage{graphicx}
\usepackage{algorithm}
\usepackage{algcompatible}
\usepackage[labelsep=endash, labelfont=bf, format=plain, textformat=simple, justification=justified, textfont=small]{caption}
\usepackage[labelfont={bf, large}, position=top, singlelinecheck=off, justification=justified, labelformat=simple, skip=0pt]{subcaption}
\DeclareCaptionLabelFormat{bold}{\textbf{#2}}
\usepackage{bm}
\usepackage[table]{xcolor}
\usepackage{stmaryrd} 
\usepackage[table]{xcolor}
\usepackage{booktabs}

\usepackage[round]{natbib}
\usepackage{hyperref} 
\usepackage[capitalise]{cleveref} 
\hypersetup{colorlinks, citecolor=blue, linkcolor=black}
\crefdefaultlabelformat{{#2#1#3}}
\crefname{chapter}{Chapter}{Chapters}
\crefname{section}{Section}{Sections}
\crefname{subsection}{Section}{Sections}
\crefname{subsubsection}{Section}{Sections}
\crefname{figure}{Figure}{Figures}
\crefname{table}{Table}{Tables}

\algnewcommand\algorithmicparam{\textbf{Param: }}
\algnewcommand\PARAM{\item[\algorithmicparam]}
\algrenewcomment[1]{\bgroup\hfill\(\triangleright\) ~#1\egroup}
\newcommand{\RR}{\mathbb{R}} 
\newcommand{\NN}{\mathbb{N}}
\newcommand{\Bcal}{\mathcal{B}}

\newcommand{\Ical}{\mathcal{I}}

\newcommand{\Ncal}{\mathcal{N}}

\newcommand{\Qcal}{\mathcal{Q}}

\newcommand{\Scal}{\mathcal{S}}

\newcommand{\Ucal}{\mathcal{U}}
\newcommand{\Vcal}{\mathcal{V}}
\newcommand{\Wcal}{\mathcal{W}}
\newcolumntype{C}[1]{>{\centering\arraybackslash}m{#1}}
\newcolumntype{L}[1]{>{\raggedright\arraybackslash}m{#1}}

\newtheorem{Lemma}{Lemma}[section]
\newtheorem{Theorem}{Theorem}[section]
\newtheorem{TheoremS}{Theorem}[subsection]

\def\abovestrut#1{\rule[0in]{0in}{#1}\ignorespaces}

\newcommand{\one}{\mathbf{1}}
\newcommand{\zero}{\mathbf{0}}
\newcommand{\beginsupplement}{%
        \setcounter{table}{0}
        \renewcommand{\thetable}{S\arabic{table}}%
        \setcounter{figure}{0}
        \renewcommand{\thefigure}{S\arabic{figure}}%
        \setcounter{algorithm}{0}
        \renewcommand{\thealgorithm}{S\arabic{algorithm}}%
        \renewcommand{\thesubsection}{\Alph{subsection}}%
        \setcounter{equation}{0}
        \renewcommand{\theequation}{S\arabic{equation}}%
     }

\newcommand {\br}[1]{\left(#1\right)}
\newcommand {\sqb}[1]{\left[#1\right]}
\newcommand {\cbr}[1]{\left\{#1\right\}}
\newcommand {\bbr}[1]{\left\llbracket#1\right\rrbracket}
\newcommand {\nm}[1]{\Arrowvert\, #1 \,\Arrowvert}
\newcommand {\abs}[1]{\left\vert\, #1 \,\right\vert}
\newcommand{\Xb}{\mathbf{X}}    
\newcommand{\zb}{\mathbf{z}}    
\newcommand{\Zb}{\mathbf{Z}}    
\newcommand{\yb}{\mathbf{y}}    
\newcommand{\ub}{\mathbf{u}}    
\newcommand{\vb}{\mathbf{v}}    
\newcommand{\Mb}{\mathbf{M}}    
\newcommand{\mb}{\mathbf{m}}    
\newcommand{\wb}{\mathbf{w}}    
\newcommand{\thetab}{\bm{\theta}}    
\newcommand{\Thetab}{\bm{\Theta}}    
\newcommand{\phib}{\bm{\theta}}    
\newcommand{\rb}{\mathbf{r}}    
\newcommand{\ab}{\mathbf{a}}    
\newcommand{\supp}{\text{supp}}

\usepackage{tikz}
\usetikzlibrary{trees, fit}
\usetikzlibrary{arrows,shapes,snakes,automata,backgrounds,petri}
\tikzstyle{source}=[circle,thick,draw=blue!75,fill=blue!20,minimum size=7mm]
\tikzstyle{sink}=[circle,thick,draw=blue!75,fill=blue!20,minimum size=7mm]
\tikzstyle{group}=[place,thick,draw=red!75,fill=red!20, minimum size=7mm]
\tikzstyle{groupgray}=[place,thick,draw=gray!75,fill=gray!60, minimum size=8mm]
\tikzstyle{empty}=[place,draw=red!0, minimum size=8mm]
\tikzstyle{var}=[rectangle,thick,draw=black!75,fill=black!20,minimum size=6mm]
\tikzstyle{varc}=[rectangle,thick,draw=darkred!75,fill=darkred!20,minimum size=6mm]
\tikzstyle{varinv}=[rectangle,thick,draw=black!0,fill=black!0,minimum size=6mm]
\tikzstyle{arrow}=[->,very thick]
\tikzstyle{arrowc}=[->,line width=0.8mm,color=darkred]
\tikzstyle{arrowb}=[->,very thick,dotted]
\definecolor{darkred}{rgb}{0.7,0,0}

\date{\today}

\begin{document}

\title{WHInter: A Working set algorithm for High-dimensional sparse second order Interaction models.}
\author{Marine Le Morvan$^{1,2,3}$ and Jean-Philippe Vert$^{1,2,3,4}$\\ \\
$^1$MINES ParisTech, PSL Research University, CBIO-Centre for Computational Biology,\\
75006 Paris, France\\
$^2$Institut Curie, 75005 Paris, France\\
$^3$INSERM, U900, 75005 Paris, France\\
$^4$Ecole Normale Sup\'erieure, Department of Mathematics and Applications, 75005 Paris, France\\ \\
\url{{marine.le_morvan, jean-philippe.vert}@mines-paristech.fr}
}

\maketitle

\begin{abstract}
Learning sparse linear models with two-way interactions is desirable in many application domains such as genomics. $\ell_1$-regularised linear models are popular to estimate sparse models, yet standard implementations fail to address specifically the quadratic explosion of candidate two-way interactions in high dimensions, and typically do not scale to genetic data with hundreds of thousands of features. Here we present WHInter, a working set algorithm to solve large $\ell_1$-regularised problems with two-way interactions for binary design matrices. The novelty of WHInter stems from a new bound to efficiently identify working sets while avoiding to scan all features, and on fast computations inspired from solutions to the maximum inner product search problem. We apply WHInter to simulated and real genetic data and show that it is more scalable and two orders of magnitude faster than the state of the art.
\end{abstract}

\section{Introduction}

In application domains where the number of features exceeds the number of available samples, sparsity-inducing regularisers have a long history of success. Genomic prediction of complex phenotypes, biomedical imaging, astronomy or finance are a few examples. In particular the least squares with $\ell_1$ regularisation, known as the LASSO \citep{Tibshirani1996Regression}, has been extensively studied. It enjoys desirable statistical properties, since the number of samples required for exact support recovery of a sparse model scales as the logarithm of the number of features, under some assumptions \citep{Wainwright2009Sharp}. It also enjoys practical advantages, notably the interpretability of the learned models and the availability of fast solvers.\\


Indeed, a lot of research effort has been devoted to accelerating solvers for sparsity constrained problems in high dimension. A central idea is to exploit the sparsity of the solution to develop algorithms that do not spend too much time on optimising coefficients that will end up being $0$. For example, safe screening rules identify features which are guaranteed to be inactive at the optimum so that their corresponding coefficients can be safely zeroed and set aside from the pool of coefficients to update \citep{ElGhaoui2012Safe,Xiang2011ST3,Xiang2012dome,Fercoq2015Mind,Wang2013EDPP, Raj2016ScreeningRulesGeneralization}. Dynamic screening rules \citep{Bonnefoy2015Dynamic} such as the GAP safe rules \citep{Fercoq2015Mind} are particularly useful since more and more coefficients can be safely zeroed while the solver approaches the optimal solution. In spite of this, safe rules tend to be conservative, thereby limiting the potential speed-up. To remedy this drawback, new working set heuristics have been proposed. Working set algorithms iteratively solve subproblems, either problems restricted to a subset of features in the primal or to a subset of constraints in the dual, until convergence. Working set methods allow to focus coefficient updates on a set of features which can be significantly smaller than that yielded by safe rules. However this comes at a cost, that of checking the optimality conditions for all features at each iteration. BLITZ \citep{Johnson2015BLITZ} is a recently proposed working set algorithm that has been shown to have state-of-the-art performance for $\ell_1$-regularised problems. Interestingly, the choice of the working sets in BLITZ can be seen as an aggressive use of the GAP safe rules \citep[as noted in][]{Massias2017WorkingSets} where the size of the working set is chosen to maximise the progress towards convergence. BLITZ can therefore be combined with the GAP safe rules (or the FLEX constraint elimination according to Johnson et al. terminology) at no cost. A direct comparison between BLITZ and the GAP safe rules by \citet{Ndiaye2017Gap} illustrates the effectiveness of the working set approach. Further developments have also focused on coordinate descent (CD) to avoid wasteful coordinate updates, which represent most of the time spent by the solver \citep{Fujiwara2016Sling, Johnson2017stingy}.\\

The problem of fitting sparse linear models with two-way interactions has also attracted attention during the past decade. By two-way interactions we mean the entry-wise multiplication between two features; this is for example important in genomics to detect possible epistasis between genes. Surprisingly, very few of these works have links with the aforementioned literature. A majority of them focus on the design of sparsity-inducing penalties which enforce heredity assumptions and apply to moderate-dimensional settings ($p < 1,000$) \citep{Radchenko2010Vanish, Bien2013lasso, Lim2015glinternet, Haris2016Family}. Heredity assumptions state that an interaction can be included in the model only if one or both of its corresponding main effects are included. We note however that \texttt{glinternet} \citep{Lim2015glinternet} was applied to higher dimensional problems and in particular to a dataset with roughly $p = 27,000$ main effects, although the size of the learned model is not specified and the running time for the experiment is not reported by the authors. Interestingly, \texttt{glinternet} uses an active set strategy. Comparatively few works have been devoted to learning sparse regression models with interactions when the number of interactions is higher. Most of them are heuristics which start by selecting main effects and then incorporate interactions generated under the heredity constraint in a possibly iterative fashion. The simplest form of such heuristics consists in fitting a sparse linear model with the main effects only, and then fitting a second sparse linear model on all previously selected main effects and their interactions. This has been used in practice for example by \citet{Wu2009TwoStagesInteractions}. Iterative refinements have been proposed where the LASSO is fit several times, and each time the set of candidate interactions considered is updated either by subsets, with the interactions between the K most relevant main effects selected at the previous fit \citep{Bickel2010Hierarchical}, or in a greedy fashion, where new interactions are included in the model as soon as a new main effect enters the LASSO path \citep{Shah2016backtracking}. In a similar vein, \citet{Hao2014iFOR} is based on a greedy model selection procedure instead of several LASSO fits. While these heuristics can deal with higher-dimensional problems than previous methods and enjoy some desirable statistical properties, they do not provide exact solutions and do not enjoy statistical properties as strong as those of the LASSO estimator.\\

An interesting link between the literature on interactions and that of solver acceleration with sparsity inducing norms has been made recently by \citet{Nakagawa2016SPP}. In the case where variables are binary or with values in $[0, 1]$, they propose an approach called Safe Pattern Pruning (SPP) which is able to provide the optimal solution of the LASSO with two-way interactions for fairly high-dimensional problems, with no heredity constraint. Typically, for a problem with 1,000 samples and 10,000 main effects, SPP can provide solutions for a grid of regularisation parameters within one or two hours on a laptop with one core. SPP relies on the recently developed GAP safe screening rules. More precisely, the authors propose a safe pattern pruning criterion that can safely discard subsets of interactions from the model to speed up convergence. The performance of SPP is however hindered by several factors. One of them is that safe screening rules can be quite conservative even in the sequential setting. This property is inherited and amplified by the SPP criterion which can lead to heavy computations. Moreover, the GAP safe rules rely on a dual feasible point which is expensive to compute especially when the number of interactions is huge.\\

Inspired by SPP and the acceleration of solvers for sparsity constrained problems we propose a scalable algorithm, WHInter, to compute the optimal solution of $\ell_1$-regularised linear problems with two-way interactions. WHInter is a working set method that efficiently delineates working sets among all interactions and main effects thanks to two contributions. First, we introduce a cheap and effective bound to rule out subsets of interactions that are guaranteed to be outside of the working set. Second, the identification of the working set among the remaining features is cast as a variant of the Maximum Inner Product Search (MIPS) problem to alleviate the afferent computational load. We find that WHInter is up to two orders of magnitude faster than SPP. For example, a problem with roughly 700 samples and 100,000 main effects can be solved for a grid of regularisation parameters in half an hour on a laptop with one core compared to more than 30 hours with SPP. This improvement in the scalability opens up new horizons in several application fields. The rest of the paper is organised as follows. In section 2, we present useful knowledge and notations used throughout the paper. In section 3 we describe in details our algorithm and our main contributions. In section 4, we evaluate WHInter on simulated datasets and finally in Section 5, we report results on a toxicogenomics prediction task.

\section{Preliminaries}
\subsection{Setting and notations}

\begin{figure}
\begin{tikzpicture}[level distance=1.5cm,
level 1/.style={sibling distance=3.9cm},
level 2/.style={sibling distance=1.25cm}]
\tikzstyle{every node}=[draw,rectangle,rounded corners=3pt, scale=0.9]

\node (Root) {$\emptyset$}
child {
    node {$ \Xb_1$} 
    child { node {$ \Xb_1 \Xb_2$} }
    child { node {$ \Xb_1 \Xb_3$} }
    child { node {$ \Xb_1 \Xb_4$} }
}
child {
    node {$ \Xb_2$} 
    child { node {$ \Xb_2 \Xb_1$} }
    child { node {$ \Xb_2 \Xb_3$} }
    child { node {$ \Xb_2 \Xb_4$} }
}
child {
    node {$ \Xb_3$} 
    child { node {$ \Xb_3 \Xb_1$} }
    child { node {$ \Xb_3 \Xb_2$} }
    child { node {$ \Xb_3 \Xb_4$} }
}
child {
    node {$ \Xb_4$} 
    child { node {$ \Xb_4 \Xb_1$} }
    child { node {$ \Xb_4 \Xb_2$} }
    child { node {$ \Xb_4 \Xb_3$} }
};
\node[draw=red,fit={(Root-1) (Root-1-1)(Root-1-2)(Root-1-3)}, label=\textcolor{red}{Branch 1}, scale=1.1] {}; 
\end{tikzpicture}
\caption{Organisation of the main effects and interactions in a tree, depicted for 4 main effects.}
\label{tree}
\end{figure}
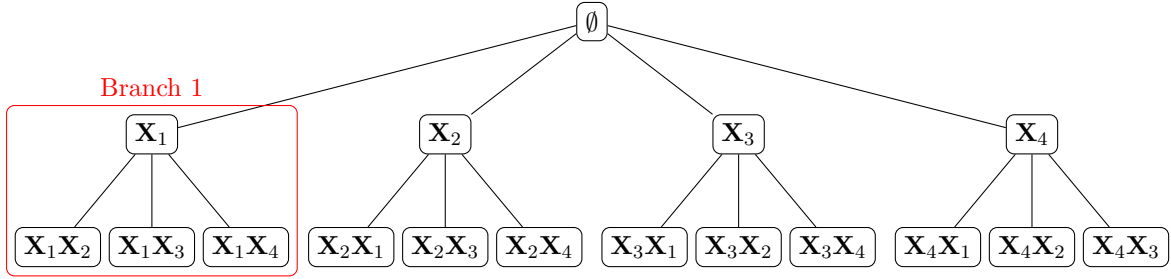

For any integer $d\in\NN$, we note $\bbr{d} = \cbr{1,\ldots,d}$ and $\one_d \in \RR^d$ the $d$-dimensional vector of $1$'s. For any vector $\ub = \br{\ub_1,\ldots,\ub_d} \in \RR^d$, we note $\nm{\ub}_1 = \sum_{i=1}^d \abs{\ub_i}$, $\nm{\ub}_2 = \br{\sum_{i=1}^d \ub_i^2}^{1/2}$, $\supp(\ub) = \cbr{i\in\bbr{d}\,:\,\ub_i \neq 0}$ and $\nm{\ub}_0 = \abs{\supp(\ub)}$. For any two vectors $\ub, \vb \in \RR^d$, $\ub \odot \vb$ is the vector of entry-wise products, i.e., $(\ub \odot \vb)_i := \ub_i \vb_i$ for $i=1,\ldots, d$. For any matrix $\Mb$, we denote by $\Mb_{i,j}$ its $(i,j)$-th entry, $\Mb_j$ its $j$-th column and by $\bm m_i$ its $i$-th row. For any $\ub\in\RR^d$ and $\Ical\subset\bbr{d}$, $\ub_\Ical = \br{\ub_i}_{i\in\Ical}$, and similarly, if $\Mb$ is a matrix with $d$ columns, $\Mb_\Ical$ is the sub-matrix with $\abs{\Ical}$ columns $\Mb_\Ical = \br{\Mb_i}_{i\in\Ical}$.

Throughout the text we consider a design matrix $\Xb \in \{0, 1\}^{n \times p}$ corresponding to $n$ samples and $p$ binary features, together with a response vector $\yb \in \RR^n$. We define an expanded design matrix $\Zb \in \{0, 1\}^{n \times D}$, with $D=p(p+1)/2$, which contains all $p$ features from $\Xb$ plus the $p(p-1)/2$ interaction features. For clarity purposes, we define a symmetric indexing function $\tau: \bbr{p}^2 \mapsto \bbr{D}$ that uniquely assigns to every main effect and interaction an index in the expanded matrix $\Zb$ such that $\Zb_{\tau(j,k)} = \Zb_{\tau(k,j)} := \Xb_j \odot \Xb_k$. In particular $\Zb_{\tau(i,i)} = \Xb_i \odot \Xb_i = \Xb_i$ represents the $i^{th}$ main effect. Since $\Xb$ is a binary matrix, the interaction feature $\Xb_j \odot \Xb_k$ corresponds to a logical AND between features $\Xb_i$ and $\Xb_j$. We organise the main effects and interactions in a simple tree as depicted in Figure~\ref{tree}  so as to reflect the property that $\forall (j, k) \in \bbr{p}^2, \Zb_{\tau(j, k)} \leq \Xb_j$ and $\Zb_{\tau(j, k)} \leq \Xb_k$. In the sequel, the set composed of a main effect and its interactions with all other main effects will be referred to as \textit{a branch} and for for any $j \in \bbr{p}$, we note branch($j$) = $\cbr{\tau(j, k): k\in \bbr{p}}$.\\

\begin{table}
\caption{\textbf{Summary of useful functions for the LASSO and logistic regression}: loss function $f_i$, its derivative $ f_i^\prime$, its Fenchel-Legendre transform $f_i^*$.}
\label{tab:losses}
\begin{center}
\begin{tabular}{cccc}
\toprule
& $f_i(u)$ & $f_i^\prime(u)$ & $f_i^*(u)$\\
\midrule
\abovestrut{0.20in}
LASSO & $\frac{1}{2} \br{\yb_i - u}^2$ & $u - \yb_i$ &  $\frac{1}{2}\br{\yb_i + u}^2 - \frac{1}{2}\yb_i^2$\\
Logistic regr. & $\log( 1 + \exp(-\yb_i u))$ & $-\frac{u}{\yb_i}\log(-\frac{u}{\yb_i}) + (1+\frac{u}{\yb_i})\log(1+\frac{u}{\yb_i})$ & $\frac{-\yb_i}{1 + \exp(\yb_i u)}$\\
\bottomrule
\end{tabular}
\end{center}
\end{table}

We consider the convex optimization problem:
\begin{equation}\label{eq:primal}
\min_{(\wb,b)\in\RR^D\times\RR} P_{\Zb, \lambda}(\wb, b) := F\br{\Zb \wb + b\one_n}  + \lambda \left\lVert \wb \right\rVert_1 := \sum_{i=1}^n f_i \left(\zb_i \wb + b\right) + \lambda \left\lVert \wb \right\rVert_1 \,,
\end{equation}
where $\lambda >0$ is a regularisation parameter and, for any $i\in\bbr{n}$, $f_i: \RR \mapsto [-\infty, +\infty]$ is a loss function parametrised by $\yb_i$ and assumed to be convex and differentiable. Table~\ref{tab:losses} provides examples of classical loss functions in classification and regression. A dual formulation of (\ref{eq:primal}) reads:
\begin{equation}\label{eq:dual}
\max_{\thetab\in\RR^n} D_{\Zb, \lambda}(\thetab) := -\sum_{i=1}^n f_i^* \left(-\thetab_i\right) \quad\text{ s.t. } \quad
\begin{cases}
\left| \Zb_{i}^\top\thetab \right| \leq \lambda \quad \forall i \in \bbr{ D} \,,\\
\one_n^\top \thetab = 0\,,
\end{cases}
\,
\end{equation}
where $f_i^*$ is the Fenchel-Legendre transform of the loss $f_i$, i.e., the function $f_i^*: \RR \mapsto [-\infty, +\infty]$ defined by $f_i^*(u) = \text{sup}_{v \in \RR}\;  uv - f_i(v)$. For the derivation of the dual problem, we refer the reader to \citet[Appendix E]{Johnson2015BLITZ}. The constraint $\one_n^\top \thetab = 0$ comes from the bias term $b\one_n$ in the primal problem (\ref{eq:primal}). We denote by $(\wb^*, b^*)$ and $\thetab^*$ a set of primal and dual optimal solutions to problems (\ref{eq:primal}) and (\ref{eq:dual}) respectively. Strong duality holds and therefore $(\wb^*, b^*)$ and $\thetab^*$ satisfy Fermat's rules \citep{Ndiaye2017Gap}:
\begin{equation}\label{eq:KKT1}
\thetab^* = -\nabla F(\Zb \wb^* + b^* \one_n)\,,
\end{equation}
and
\begin{equation}\label{eq:KKT2}
\forall i \in \bbr{D}, \quad \Zb_{i}^\top \thetab^*\in
\begin{cases}
\left\{-\lambda, \lambda \right\} &\text{ if } \wb^*_{i} \ne 0 \,,\\
\left[ -\lambda, \lambda \right] &\text{ if } \wb^*_{i} = 0\,.
\end{cases}
\end{equation}

\subsection{Basic working set algorithm}
A general strategy to solve (\ref{eq:primal}) is to follow a \emph{working set} approach, as summarised in Algorithm~\ref{alg:working_set}. At each iteration, it solves (\ref{eq:primal}) restricted to a small subset of features $\Wcal$ called the working set. $\Wcal$ is typically chosen as the set of features that violate the optimality condition (\ref{eq:KKT2}) at the current iteration. In the sequel, we will call such features \textit{violating features}, and the branches which contain at least one violating feature will be called \textit{violating branches}. The algorithm converges when no violating feature remains, which occurs in a finite number of iterations as shown in \citet{Kowalski2011Accelerating}. When the number of interaction features runs into the billions, Algorithm~\ref{alg:working_set} is not tractable since the delineation of the working set (line 3 in Alg.~\ref{alg:working_set}) requires $O(p^2n)$ operations at each iteration.

\begin{algorithm}
\caption{Working set algorithm}
\label{alg:working_set}
\begin{algorithmic}[1]
\REQUIRE $\Zb \in \{0,1\}^{n \times p}, \yb \in \RR^n$, $\lambda>0$
\ENSURE  $\wb^*, b^*$
\STATE Set $\thetab \leftarrow -\nabla F(\zero_n)$, $\Wcal=\emptyset$.\Comment{Initialisation}
\WHILE {true}
\STATE $\Wcal^\prime = \left\{i \in \bbr{D}: \left|\Zb_{i}^\top \thetab \right| \geq \lambda \right\}$ \Comment{Update the working set}
\STATE \textbf{if} $\max_{i\in\Wcal'} \abs{\Zb_i^\top \theta} \leq \lambda$ \textbf{then} Break \textbf{else} $\Wcal \leftarrow \Wcal^\prime$
\STATE  $\wb_{\Wcal}^*, b^*\leftarrow \underset{\wb_{\Wcal}, b}{\text{argmin }} P_{\Zb_\Wcal, \lambda} (\wb_{\Wcal}, b)$ \Comment{Solve subproblem}
\STATE  $\thetab \leftarrow -\nabla F(\Zb_\Wcal \wb_{\Wcal}^* + b^*\one_n)$.
\ENDWHILE
\end{algorithmic}
\end{algorithm}

\section{The WHInter algorithm}
\subsection{Overview}
WHInter is a working set algorithm that follows the general scheme of Algorithm~\ref{alg:working_set} but implements an efficient strategy to delineate the working set among all main effects and interactions. It is described in Algorithm~\ref{alg:WHInter_core}. The identification of the working set (line 3 in Algorithm~\ref{alg:working_set}) corresponds to lines 11-18 in Algorithm~\ref{alg:WHInter_core}. Instead of scanning through all features to build the working set, WHInter first identifies branches that are guaranteed to contain no violating feature. These branches are identified via the evaluation of a \textit{branch bound} $\eta(\Xb_j, \Thetab^{ref}_j, \phib, \mb_j^{ref})$ (line 13) which is described in Section \ref{subsec:BB}. The branch bound is cheap to evaluate since it solely depends on main effects and not on their numerous interactions. Moreover, it is designed to efficiently rule out branches thanks to the exploitation of the shared structure among features in a branch, as well as the correlation among dual variables for two sufficiently close points in the optimisation path. In cases where a branch cannot be ruled out, features in the branch are considered one by one to build the working set, which is very computationally expensive. In order to reduce this cost, we cast the problem as a variant of the Maximum Inner Product Search (MIPS) problem, which is described in Section \ref{subsec:MIPS}. If no violating feature is identified then the algorithm has converged. Otherwise, a new candidate solution is obtained by solving problem~(\ref{eq:primal}) restricted to the features in the working set, and the process is repeated until no violating feature remains. While any solver can be used to solve the restricted problem, we implemented in WHInter a coordinate descent approach with safe pruning.


\begin{algorithm}
\caption{WHInter}
\label{alg:WHInter_core}
\begin{algorithmic}[1]
\REQUIRE $\Xb \in \{0, 1\}^{n \times p},\, \yb \in \RR^n$,\, $\lambda_1 > \dots > \lambda_T.$
\ENSURE  $(\Wcal, \wb_\Wcal^*, b^*)_t$ for each $\lambda_t$

\STATE $\thetab \leftarrow -\nabla F(\zero_n)$
\FOR{$j$ in $\bbr{p}$}
\STATE $\Thetab_j^{ref} \leftarrow \bm \theta$
\ENDFOR
\STATE  $\Wcal, \mb^{ref} \leftarrow \texttt{update\_W}(\Xb, \phib, \bbr{p}, \lambda_1, \emptyset)$ \Comment{See Section~\ref{subsec:MIPS}}

\FOR{$t=1$ to $T$} 

\STATE $\wb_{\Wcal}^*, b^*\leftarrow \underset{\wb_{\Wcal}, b}{\text{argmin }} P_{\Zb_\Wcal, \lambda_t} (\wb_{\Wcal}, b)$ \Comment{Pre-Solve}
\STATE  $\thetab \leftarrow -\nabla F(\Zb_\Wcal \wb_{\Wcal}^* + b^*\one_n)$.
\STATE $\Wcal, \bm m^{ref} \leftarrow \texttt{clean\_W}(\Wcal, \lambda_t, \thetab, \Thetab^{ref}, \bm m^{ref}$)
\WHILE {true}
\STATE $\Vcal \leftarrow \emptyset$ \Comment{Identify violated branches}
\FOR{$j$ in $\bbr{p} $}
\IF{$\eta(\Xb_j, \Thetab^{ref}_j, \phib, \mb_j^{ref}) >\lambda_t$} \Comment{See Section~\ref{subsec:BB}}
\STATE $\Vcal \leftarrow\Vcal \cup \{j\}$
\STATE $\Thetab^{ref}_j \leftarrow \phib$
\ENDIF
\ENDFOR
\STATE $\Wcal^\prime, \mb^{ref}_{\Vcal} \leftarrow {\texttt{update\_W}}(\Xb, \phib, \Vcal, \lambda_t, \Wcal)$ \Comment{See Section~\ref{subsec:MIPS}}
\STATE \textbf{if} $\max_{i\in\Wcal'} \abs{\Zb_i^\top \theta} \leq \lambda$ \textbf{then} Break \textbf{else} $\Wcal \leftarrow \Wcal^\prime$
\STATE $\wb_{\Wcal}^*, b^*\leftarrow \underset{\wb_{\Wcal}, b}{\text{argmin }} P_{\Zb_\Wcal, \lambda_t} (\wb_{\Wcal}, b)$ \Comment{Solve subproblem}
\STATE  $\thetab \leftarrow -\nabla F(\Zb_\Wcal \wb_{\Wcal}^* + b^*\one_n)$.

\STATE $\Wcal, \bm m^{ref} \leftarrow \texttt{clean\_W}(\Wcal, \lambda_t, \thetab, \Thetab^{ref}, \bm m^{ref}$)
\ENDWHILE
\STATE $(\Wcal, \wb_\Wcal^*, b^*)_k \leftarrow (\Wcal, \wb_\Wcal^*, b^*)$
\ENDFOR

\vspace{0.5em}
\hrule
\vspace{0.5em}


\STATE \textbf{function} \texttt{clean\_W}($\Wcal, \lambda, \thetab, \Thetab^{ref}, \bm m^{ref}$)
\STATE \hspace{1em} \textbf{for} $i$ in $\Wcal$ \textbf{do}
\STATE \hspace{2em} \textbf{if} $\left|\Zb_i^\top \thetab\right| < \lambda$ \textbf{then}
\STATE \hspace{3em} Remove $\{i\}$ from $\Wcal$
\STATE \hspace{3em} \textbf{for} $b$ in branch($i$) \textbf{do}
\STATE \hspace{4em} \textbf{if} $\mb_b^{ref} < \left|\Zb_i^\top \Thetab^{ref}_b \right|$ \textbf{then} $\mb_b^{ref} \leftarrow \left|\Zb_{i}^\top\Thetab^{ref}_b \right|$
\STATE \hspace{1em} \textbf{return} $\Wcal$, $\bm m^{ref}$

\end{algorithmic}
\end{algorithm}

\subsection{The Branch bound $\eta$}
\label{subsec:BB}
As WHInter iterates, it produces candidate solutions $(\wb^*, b^*)$ and corresponding dual variables $\thetab$ (lines 20 and 21 of Algorithm \ref{alg:WHInter_core}). For two sufficiently close iterations, or for two problems with sufficiently close regularisation parameters, the candidate solutions are likely to be \textit{close} to one another, as well as the corresponding dual variables. WHInter exploits this intuition to speed up the identification of the working set from an iteration to another or from one problem to another. The following results relate the criteria used to identify the working set (line 3 of Algorithm~\ref{alg:working_set}) for two distinct dual variables.
\begin{Lemma}
\label{lem:bound}
For any $\Xb\in\cbr{0,1}^{n\times p}$, $\vb\in\RR_+^n$, $\phib_1, \phib_2 \in \RR^n$, $j\in\bbr{p}$, $\Ical\subset\bbr{p}$ and $\alpha\in\RR$, the following holds:
\begin{equation}\label{eq:lem1}
\underset{k \in \Ical}{\max} \abs{ \phib_2^\top (\vb \odot \Xb_k)} \leq \abs{\alpha} \underset{k \in \Ical}{\max} \abs{ \phib_1^\top (\vb \odot \Xb_k)} + \zeta(\phib_2 - \alpha \phib_1, \vb)\,,
\end{equation}
where
$$
\forall (\ub, \vb) \in \RR^n \times \RR_+^n\,,\quad \zeta(\ub, \vb) =  \max\br{ \displaystyle{\sum_{i : \ub_i > 0} \ub_i \vb_i }, \displaystyle{-\sum_{i : \ub_i < 0} \ub_i \vb_i} } \,.
$$
\end{Lemma}
The proof of Lemma~\ref{lem:bound} is postponed to Appendix~\ref{subsec:prooflemma}. It is based on the decomposition $\phib_2 = \alpha \phib_1 + \left( \phib_2 - \alpha \phib_1 \right)$, and exploits the tree structure among features in a branch.
To exploit Lemma~\ref{lem:bound} in WHInter, we define for $\alpha\in\RR$ the function
\begin{equation}\label{eq:eta}
\forall\br{\vb,\thetab_1,\thetab_2,m} \in \RR_+^n \times \RR^n\times \RR^n \times \RR\,,\quad\eta_\alpha\br{\vb,\thetab_1,\thetab_2,m} = \abs{\alpha} m + \zeta\br{\phib_2 - \alpha \phib_1, \vb} \,,
\end{equation}
and we maintain an active set $\Wcal \subset \bbr{D}$, a matrix $\Thetab^{ref} \in \RR^{n\times p}$ that contains \textit{reference dual variables} $\Thetab^{ref}_j \in \RR^n$ for each branch $j\in \bbr{p}$, and the vector $\mb^{ref}\in\RR^p$ defined by:
\begin{equation}\label{eq:mref}
\forall j\in\bbr{p},\quad \mb^{ref}_j = \underset{k \in \bbr{p} : \tau(j,k)\notin \Wcal}{\max} \abs{ \Zb_{\tau(j,k)}^\top \Thetab_j^{ref}} \,.
\end{equation}
We now state our pruning theorem which allows to identify branches which are guaranteed to not contain any violating feature (line 13 of algorithm~\ref{alg:WHInter_core}):
\begin{Theorem}[Branch pruning]
\label{th:eta}
For any $\Thetab^{ref} \in \RR^{n\times p}$, $\Wcal\subset\bbr{p}$, $j\in\bbr{p}$, let $\mb^{ref}_j \in \RR_+$ be given by (\ref{eq:mref}). Then for any $\phib \in \RR^n$, $\alpha \in \RR$ and $\lambda>0$, if
\begin{equation}\label{eq:eta1}
\eta_{\alpha}\br{\Xb_j,\Thetab_j^{ref}, \thetab, \mb_j^{ref}} < \lambda \,,
\end{equation}
then any feature from branch $j$ that belongs to the working set $\left\{i \in \bbr{D}: \left|\Zb_{i}^\top \thetab \right| \geq \lambda \right\}$ is already in $\Wcal$. This holds in particular if
\begin{equation}\label{eq:eta2}
\eta_{min}\br{\Xb_j,\Thetab_j^{ref}, \thetab, \mb_j^{ref}}  := \min_{\alpha\in\RR }\eta_{\alpha}\br{\Xb_j,\Thetab_j^{ref}, \thetab, \mb_j^{ref}} < \lambda \,.
\end{equation}
\end{Theorem}
\begin{proof}
Take $\Ical = \cbr{k \in \bbr{p} : \tau(j,k)\notin \Wcal}$, $\vb=X_j$, $\thetab_1 = \Thetab^{ref}_j$ and $\thetab_2 = \thetab$ in Lemma~\ref{lem:bound}. Then if (\ref{eq:eta1}) holds, we deduce from (\ref{eq:lem1}) that
$$
\underset{k \in \bbr{p} : \tau(j,k)\notin \Wcal}{\max} \abs{ \Zb_{\tau(j,k)}^\top \thetab} < \lambda \,.
$$
This shows that there is no feature $i$ in branch $j$ such that $\left|\Zb_{i}^\top \thetab \right| \geq \lambda$ \textit{and} $i$ is not already in $\Wcal$. The fact that for fixed arguments, the function $\alpha\rightarrow\eta_\alpha$ has a minimum $\alpha^*\in\RR$ is shown in Appendix \ref{subsec:etamin}, along with with an algorithm to compute it in $O\br{\nm{\Xb_j}_0 \ln \nm{\Xb_j}_0}$ operations. Since the statement is true for any $\alpha$, it is \emph{a fortiori} true for $\alpha^*$.
\end{proof}

Theorem~\ref{th:eta} provides criteria (\ref{eq:eta1}) and (\ref{eq:eta2}) that can be computed for each branch $j$, and which if satisfied allow to skip the search for violating variables in the branch. Importantly, the features that are already in the working set $\Wcal$ are not taken into account to compute the criterion for a given branch. This subtlety allows to rule out branches even if they already contain features that were previously incorporated in the working set. Note that the reference dual variable for branch $j$, i.e, $\Thetab_j^{ref}$, is kept unchanged as long as branch $j$ is pruned, and is otherwise updated to the latest dual variable (line 15 of Algorithm~\ref{alg:WHInter_core}). As $\mb^{ref}_j$depends on the reference dual variable instead of the current one, it is solely reevaluated each time the reference residual is updated (line 18 of Algorithm~\ref{alg:WHInter_core}) or when a feature from branch $j$ leaves the working set (line 22 of Algorithm~\ref{alg:WHInter_core}) .


Criterion (\ref{eq:eta2}) is the most stringent one, and therefore the most efficient one to prune branches, but it takes $O\br{\nm{\Xb_j}_0 \ln \nm{\Xb_j}_0}$ operations to compute. In order to balance computational complexity of the bound with its efficacy to prune branches, criterion (\ref{eq:eta1}) can be used as an alternative for a specific $\alpha$ value. One simple choice is to just take $\alpha=1$, which leads to the criterion
\begin{equation}\label{eq:alpha1}
\eta_{1}\br{\Xb_j,\Thetab_j^{ref}, \thetab, \mb_j^{ref}} = \mb_j^{ref} + \zeta\br{\Thetab_j^{ref} - \thetab, \Xb_j} < \lambda\,.
\end{equation}
Alternatively, a simple heuristic to expect a more efficient pruning is to choose an $\alpha$ that minimises $\nm{\br{\thetab - \alpha \Thetab_j^{ref}}\odot \Xb_j}_2$, i.e,
\begin{equation}\label{eq:alpha2}
\alpha_{\ell_2} = \frac{\phib^\top \br{\Thetab^{ref}_j\odot\Xb_j}}{\nm{ \Thetab^{ref}_j\odot\Xb_j}_2^2}\,.
\end{equation}
$\eta_{\alpha_{\ell_2}}$ is expected to be more effective than $\eta_1$ since it is reasonable to expect that $\zeta\br{\thetab - \alpha_{\ell_2} \Thetab_j^{ref} , \Xb_j}$ is smaller than  $\zeta\br{\thetab - \Thetab_j^{ref} , \Xb_j}$. Overall, computing criterion (\ref{eq:eta2}) for $\alpha=1$ as in (\ref{eq:alpha1}), or for $\alpha=\alpha_{\ell_2}$ as in (\ref{eq:alpha2}), is an $O(\nm{\Xb_j})$ operation. Since computing $\zeta(\phib - \alpha \Thetab^{ref}_j, \Xb_j)$ for a fixed $\alpha$ is also a $O(\nm{\Xb_j})$ computation, the total cost of identifying branch $j$ as violated is $O(\nm{\Xb_j})$ for criterion (\ref{eq:alpha1}), compared to $O\br{\nm{\Xb_j}_0 \ln \nm{\Xb_j}_0}$ for criterion (\ref{eq:eta2}). In Algorithm~\ref{alg:WHInter_core}, the notation $\eta$ refers to a user-defined function among $\eta_1, \eta_{\alpha_{\ell_2}}$ or $\eta_{min}$.

\subsection{Updating the working set}
\label{subsec:MIPS}
When some branches $\Vcal\subset\bbr{p}$ cannot be pruned, the simultaneous updates of the working set $\Wcal$ and of $\mb_\Vcal^{ref}$ requires scanning through all features in the branches $\Vcal$  (lines 5 and 18 in Algorithm~\ref{alg:WHInter_core}). In what follows we discuss strategies to make these updates efficient. For that purpose, let us first notice that:
\begin{align*}
\forall j, k \in \bbr{p} \,, \left|\Zb_{\tau(j, k)} ^\top \thetab \right| &= \left|\left(\Xb_j \odot \Xb_k \right)^\top \thetab\right|\\
&= \left| \left( \Xb_j \odot \thetab\right)^\top \Xb_k \right|\\
&= \left|\bm Q_j^\top \Xb_k \right| \,,
\end{align*}
where for any $j \in \bbr{p}\,,\bm Q_j = \bm X_j \odot \thetab$. This allows us to write the updates of $\Wcal$ and $\bm m^{ref}_\Vcal$ as:
\begin{equation}
\label{eq:updates}
\begin{cases}
\Wcal^\prime& = \Wcal \cup \cbr{\tau(j, k): j \in \Vcal, k\in \bbr{p}, \left|\bm Q_j^\top \bm X_k \right|\geq \lambda} \,,\\
\bm m_j^{ref}& = \underset{k: \, \left|\bm Q_j^\top \bm X_k\right|< \lambda}{\max} \left|\bm Q_j^\top \bm X_k\right|,\, \forall j \in \Vcal \,.
\end{cases}
\end{equation}
This highlights the fact that the updates of the working set $\Wcal$ and of $\bm m_\Vcal^{ref}$ can be cast as particular variants of the Maximum Inner Product Search (MIPS) problem. MIPS aims at finding a vector in a database of probes which maximises the inner product with a given query vector.
If we consider $\bm X$ as a set of probes, and $\bm Q_j$ as a query, then (\ref{eq:updates}) is a variant of MIPS where (i) the set of probe vectors satisfies some constraints and is not known upfront and (ii) the problem is a maximum \textit{absolute} inner product search. The update of $\Wcal$ involves what is sometimes referred to as \textit{above}-$\lambda$-MIPS problems where again, maximum \textit{absolute} inner products  are considered.\\

The interest of casting these updates as variants of MIPS problems is to exploit the ideas developed in the literature for solving these problems efficiently. \citet{Teflioudi2016LEMP} and \citet{Fontoura2011evaluation} give good overviews of MIPS solvers developed for recommender systems and information retrieval applications respectively. In both cases, the proposed methods rely on two main ideas: (i) adequate indexing techniques or data structures and (ii) pruning criteria which allow to not compute all inner products entirely. Since none of these methods can directly be applied to problem (\ref{eq:updates}) because of its specificities, we propose an appropriate algorithm based on a simple inverted index approach, which we will refer to as $IL$, and which exploits the sparsity of the problem. Another option would be to leverage pruning techniques. We detail such an attempt in Appendix~\ref{subsec:MIPS1}. However, since our preliminary results with the pruning technique were not conclusive compared to IL on the simulated and real data, we will only focus on the inverted index approach below.

\begin{algorithm}
\caption{\texttt{update\_W}}
\label{alg:MIPS2}
\begin{algorithmic}[1]
\REQUIRE $\Xb \in \{0, 1\}^{n \times p}, \,\phib \in \RR^n, \, \Qcal \subset \bbr{p}, \, \lambda \in \RR, \, \Wcal \subset \bbr{D}$
\ENSURE $\Wcal,\, \mb^{ref}$
\FOR{$j \in  \Qcal$}
\STATE{Initialise an array $\ab$ of size $p$ to zero.}
\FOR{\textbf{each} $i$ in $\supp(\Xb_j)$}
\FOR{\textbf{each} $k$ in supp($\bm x_i$)}
\STATE $\ab_k = \ab_k + \phib_i$ 	
\ENDFOR
\ENDFOR
\FOR{\textbf{each} $k$ s.t. $\ab_k \ne 0$}
\STATE \textbf{if}  $\mb^{ref}_j < \ab_k < \lambda$ \textbf{then} set $\mb^{ref}_j = \ab_k$
\STATE \textbf{if}  $\ab_k \geq \lambda$ and $\tau(j, k) \notin \Wcal$ \textbf{then} add $\tau(j, k)$ to $\Wcal$
\ENDFOR
\ENDFOR
\STATE \textbf{return} $\Wcal, \bm m^{ref}$
\end{algorithmic}
\end{algorithm}

$IL$ is detailed in Algorithm \ref{alg:MIPS2}. The inverted indices consist of $n$ lists, one for each dimension, where each list supp($\bm x_i$) records the indices of the features in $\Xb$ which have a non-zero element for the $i^{th}$ dimension. These inverted lists can be computed once for all when WHInter starts and be reused for all MIPS problems, and therefore building the inverted lists requires a negligible additional computational cost. Algorithm (\ref{alg:MIPS2}) computes inner product following a \textit{term-at-a-time} (TAAT) scheme \citep{Fontoura2011evaluation}, i.e, the inner products are accumulated simultaneously across probes and the contribution of the $i^{th}$ dimension to the inner products is entirely processed before moving to the next one.

\section{Simulation study}

\begin{figure*}
\centering
\begin{subfigure}[t]{0.315\linewidth}
	\caption{}\label{fig:WHInter_p}
	\includegraphics[scale=0.58]{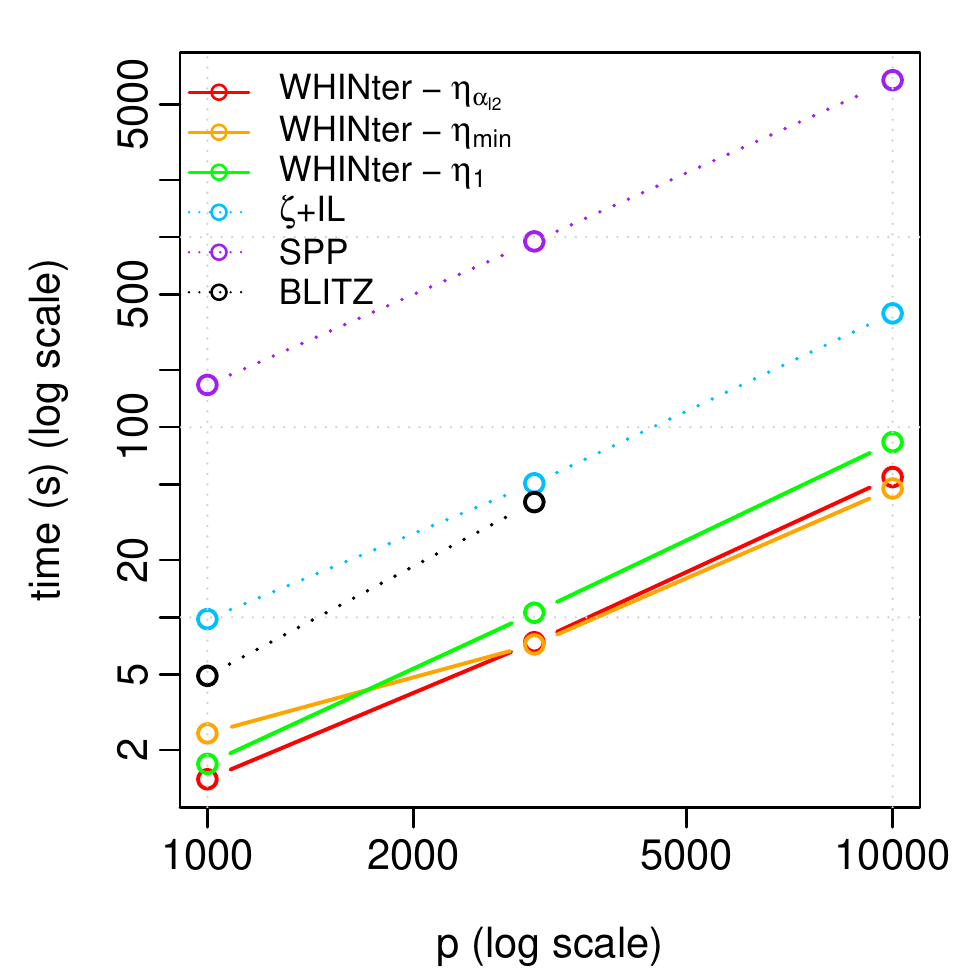}
\end{subfigure}
\;
\begin{subfigure}[t]{0.315\linewidth}
	\caption{}\label{fig:WHInter_n}
	\includegraphics[scale=0.58]{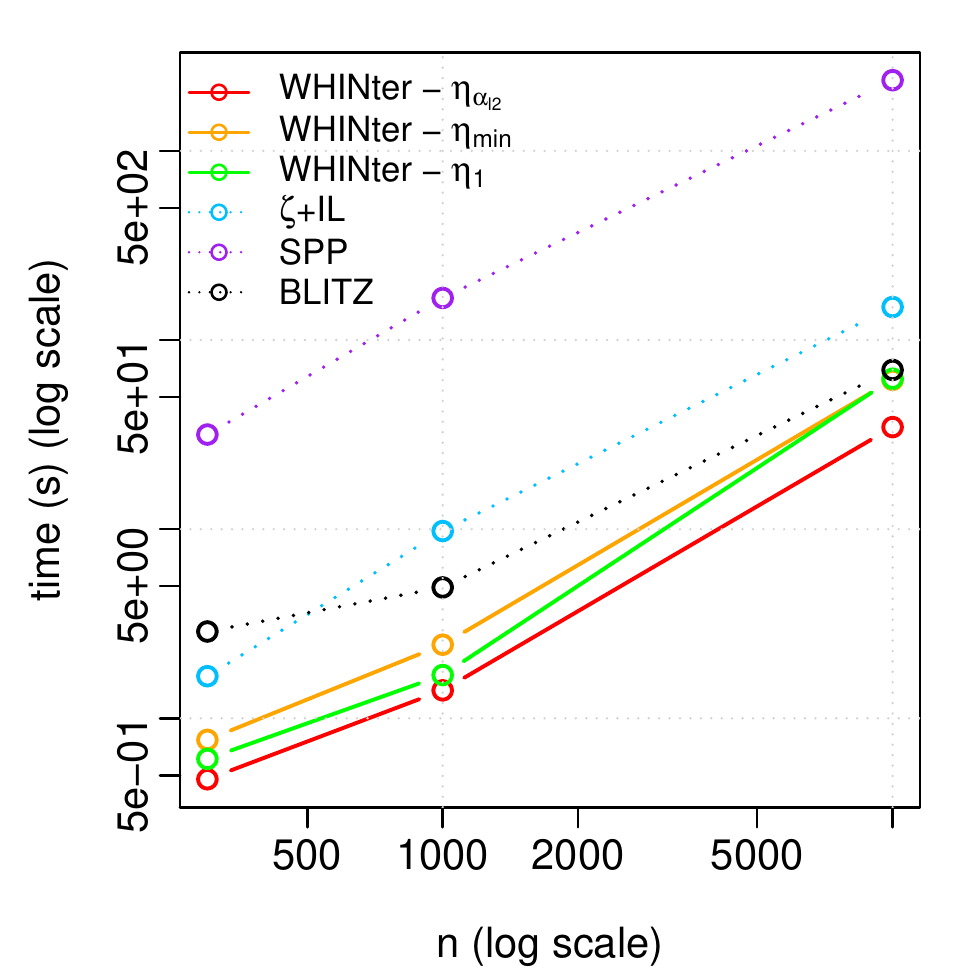}	
\end{subfigure}
\;
\begin{subfigure}[t]{0.315\linewidth}
	\caption{}\label{fig:branches1000}
	\includegraphics[scale=0.58]{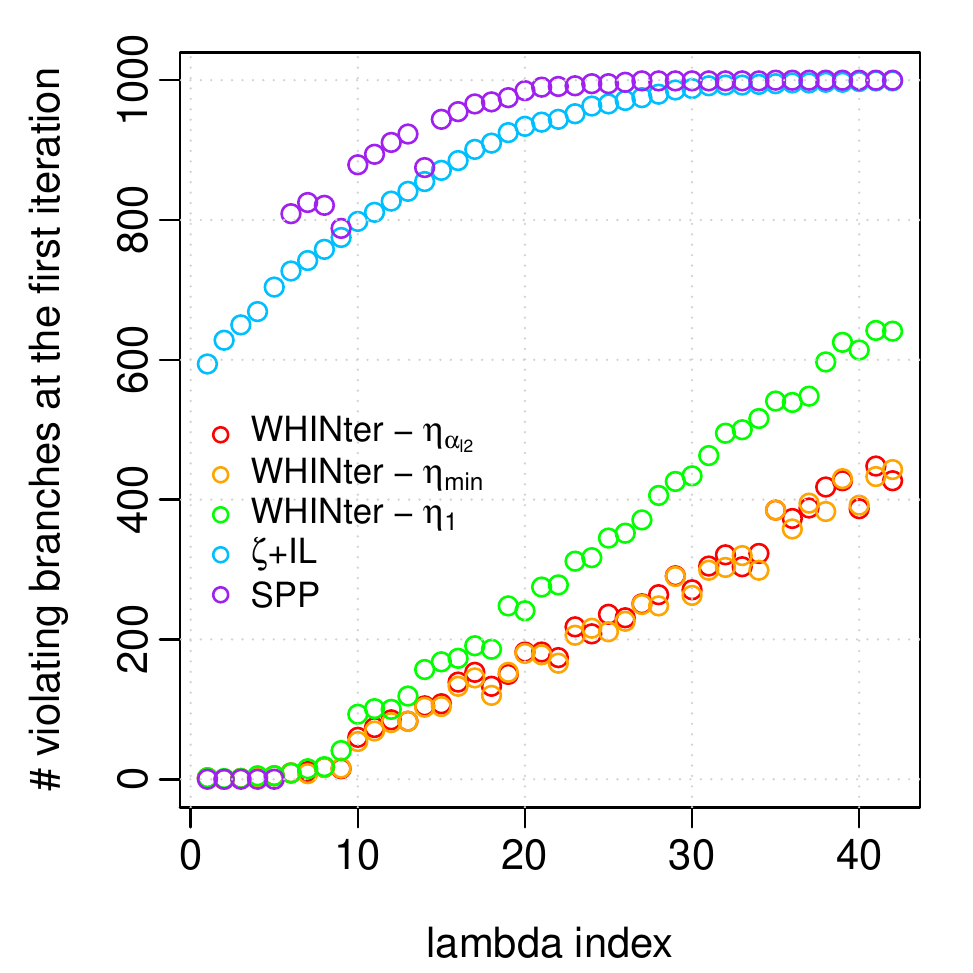}	
\end{subfigure}
\caption{\textbf{Performance comparison on simulated data for an entire regularisation path}. Comparison of WHInter with three branch pruning criteria $\eta\in\cbr{\eta_{\alpha_2},\eta_{min},\eta_1}$ to $\zeta+IL$, SPP and BLITZ.  (\subref{fig:WHInter_p}) Time in seconds for $n=1\times10^3$ fixed and $p$ varied. (\subref{fig:WHInter_n}) Time in seconds for $p=1\times10^3$ fixed and $n$ varied. (\subref{fig:branches1000}) number of branches that are \textit{not} pruned at the first iteration, as a function of $\lambda$, for $n=p=1\times10^3$.}
\label{fig:time_comparison}
\end{figure*}

We first test the performances of WHInter on synthetic LASSO datasets. We assess the performances of the different branch pruning bounds presented in \ref{subsec:BB}, i.e, $\eta_{min}$, $\eta_1$ and $\eta_{\alpha_{\ell_2}}$, and further compare WHInter to a working set method that uses the bound $\zeta(\phib, \Xb_j)$ instead of $\eta_\alpha$, but is otherwise equivalent to WHInter. We refer to this method as $\zeta + IL$. It is expected to prune less branches than WHInter but does not require to maintain $\bm m^{ref}$. We also compare WHInter to SPP \citep{Nakagawa2016SPP} and BLITZ \citep{Johnson2015BLITZ}. In our experiments, we use a slightly modified, more efficient version of the code provided by the authors of SPP (cf Appendix~\ref{subsec:breadthfirst}). As for BLITZ, since the method is not tailored for interaction problems, we first compute the matrix $\bm Z$ which is fed as input to BLITZ. For this reason we could not solve problems when $p$ is too large (e.g., $p=1\times10^4$ in the simulations) since, even in sparse format, storing $\bm Z$ requires too much memory. Importantly, the performances reported for BLITZ do not include the time required to compute $\bm Z$ from $\bm X$, which clearly advantages BLITZ compared to the other methods.

We simulate five datasets $\Xb \in \cbr{ 0, 1 }^{n \times p}$ with varying number of features and samples: three datasets with $p=1\times10^3$ fixed and $n \in \left\{3\times10^2, 1\times10^3, 1\times10^4\right\}$, and two more with $n=1\times10^3$ fixed and $p \in \left\{3\times10^3, 1\times10^4\right\}$. The features are drawn from a Bernoulli distribution with parameter $q \in [0.1, 0.5]$ itself drawn from a uniform distribution $\Ucal_{[0.1, 0.5]}$. We then randomly pick a set $\Scal$ of 100 features among the main effects and interactions and compute the response as $\yb = \Zb_\Scal \wb_\Scal^*$ where $\wb_\Scal^* \sim \Ncal(\zero_{|\Scal|}, I_{|\Scal|})$. In all experiments, the LASSO is solved for a sequence $(\lambda_t)_{t\in \bbr{T}}$, $T=100$, logarithmically spaced between $\lambda_{max}$ and $\max(0.01\lambda_{max}, \,\lambda^\prime)$ where $\lambda_{max}$
is the largest value of $\lambda$ for which at least one feature is selected, and $\lambda^\prime$ is the first $\lambda_{i}$ for which 150 features or more are selected in the model. For all methods, the time to compute $\lambda_{max}$ is included in the total time required to solve the regularisation path. In WHInter, $\lambda_{max}$ can easily be deduced from the initialisation of $\mb^{ref}$ since $\lambda_{max} = \max_{j \in \bbr{p}} \bm m_j^{ref}$. All algorithms are implemented in C++ and compiled with the \texttt{-O3} optimisation flag. The experiments are run on a 64-bit machine with Intel Core i7 Processor 2.5 GHz, 16GB of memory and 6MB of cache.

Results are shown in Figure~\ref{fig:time_comparison}. For $n=1\times10^3$ (Figure~\ref{fig:WHInter_p}), LASSO solutions are computed for 42, 32 and 28 values of $\lambda$ for $p=1\times10^3, p=3\times10^3$ and $p=1\times10^4$ respectively. In these cases smaller values of $\lambda$ result in model sizes exceeding 150 features. For the remaining settings where $p=1\times10^3$ and $n=3\times10^2$ or $n=1\times10^4$ (Figure~\ref{fig:WHInter_n}), LASSO solutions are computed for 34 and all 100 values of $\lambda$ between $\lambda_{max}$ and $0.01\lambda_{max}$, respectively. We checked that all methods return the exact same support.

In all settings, WHInter is the fastest method. Its better performance compared to $\zeta + IL$ highlights the benefit of using reference dual variables even if it implies to maintain $\bm m^{ref}$. The results also show the importance of $\alpha$, since WHInter with $\eta_{\ell_2}$ is always better ($\times 1.2$ to  $\times 1.8$) than WHInter with $\eta_1$ for example. Figure~\ref{fig:branches1000} confirms that the choice of $\alpha$ has an impact on the pruning efficiency and consequently on the performance. It shows, however, that on this experiment $\eta_{min}$ does not allow to prune many more branches than $\eta_{\ell_2}$. This explains why $\eta_{\ell_2}$ tends to outperform $\eta_{min}$, notably for large $n$, since the higher computational complexity of $\eta_{min}$ does not sufficiently enhance the pruning. We also notice that SPP is the slowest algorithm, and in particular $\zeta + IL$ is $\times 17$ faster than SPP on average.  This speed-up is mostly explained by the fact that  $\zeta + IL$ relies on inverted lists to update the working set while SPP identifies the safe set naively. 
Overall, WHInter offers a signifiant speed-up of two orders of magnitude or more compared to its safe screening counterpart.

\section{Results on real world data}

\begin{figure}
\centering
\includegraphics[scale=0.58]{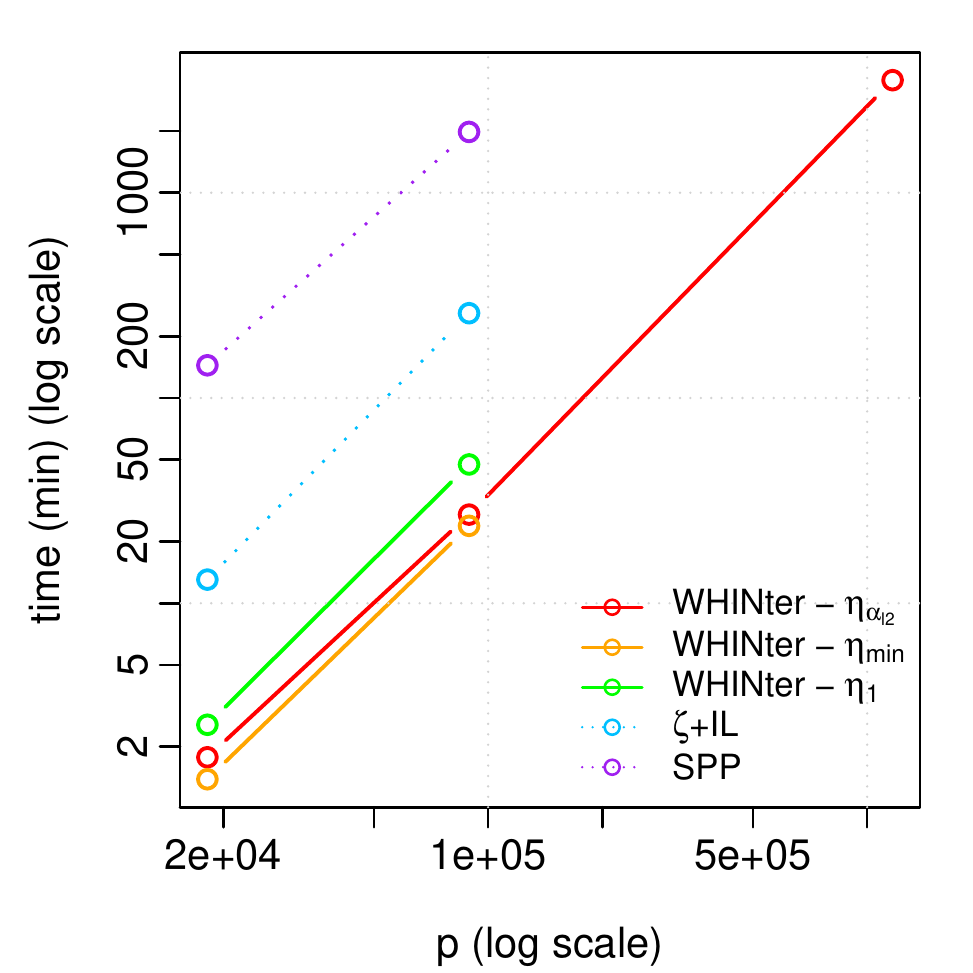}
\caption{\textbf{Performance comparison on SNPs data for an entire regularisation path}. The $y$-axis reports the total time (in minutes) required to compute the LASSO path for chromosome 22 (around 20,000 SNPs), chromosome 1 (around 90,000 SNPs) and the whole genome (around 1.2 million SNPs).}
\label{fig:WHInter_realdata1}
\end{figure}

We now illustrate the performance of the different algorithms on a real-world problem, where we want to predict the cytotoxic response of 884 lymphoblastoid cell lines split into a train ($n=620$) and a test ($n=264$) set, and characterized by about $1.2\times 10^6$ single nucleotide polymorphisms (SNP) that represent their genotypes. The data was released as part of the Dialogue on Reverse Engineering Assessment and Methods 8 (DREAM 8) toxicogenetics challenge \citep{Eduati2015Prediction}. We encode the SNP data as a binary matrix were $1$ stand for the presence of a minor allele on one or both copies of the chromosomes. As preprocessing we removed SNP with less than $5\%$ of $1$'s and corrected the data for population structure as in \citet{Price2006Principal}.
To focus on problems of increasing scales, we first considered the SNPs of the smallest chromosome only (chr. 22), then of the largest only (chr. 1) and finally of all chromosomes together. This leads to train matrices with $n=620$ and $p=18,168$ SNPs for chromosome 22, $p=89,027$ SNPs for chromosome 1 and $p=1,166,836$ SNPs for the whole genome. We consider a sequence of 100 regularisation parameters $\lambda$ logarithmically spaced between $\lambda_{max}$ and $0.01\lambda_{max}$, and by default stop computations as soon as 150 features or more are selected. This occurs after the $12^{th}$, the $11^{th}$ and the $9^{th}$ value of $\lambda$ for chromosome 22, chromosome 1 and all chromosomes respectively. The time required to compute the regularisation paths are shown in Fig. \ref{fig:WHInter_realdata1}.

The relative performances of the methods are the same as for the simulations. $\eta_{\alpha_{\ell_2}}$ provides a $\times 1.4$ (resp. $\times 1.8$) speed up compared to using $\eta_1$ for chromosome $22$ (resp.  chr. $1$). and compared to SPP, there is a $\times 81$ (resp. $\times 73$) speed up for chromosome 22 (resp chr. 1). In the case of the whole genome, we only ran WHInter with $\eta_{\alpha_{\ell_2}}$ which takes two days and a half. While this can seem a lot, we recall that this corresponds to a problem with roughly 680 \textit{billion} features. We did not run other methods on the whole genome since most of them are expected to take too long.

Out of curiosity, we also obtained preliminary results concerning the predictive performance of WHInter compared to a LASSO with no interactions on such high-dimensional problems. The results, presented in Figure~\ref{fig:WHInter_realdata} , suggest that interactions are relevant predictors for this data. For the chromosomes 1 and 22 taken independently, the predictive accuracy of WHInter is better than that of the simple LASSO for almost every value of $\lambda$. By contrast, for the whole genome, the LASSO clearly performs better, which may underline statistical issues due to the huge number of variables in this case \citep{Donoho2009Phase}.

\begin{figure*}
\centering
\includegraphics[scale=0.53]{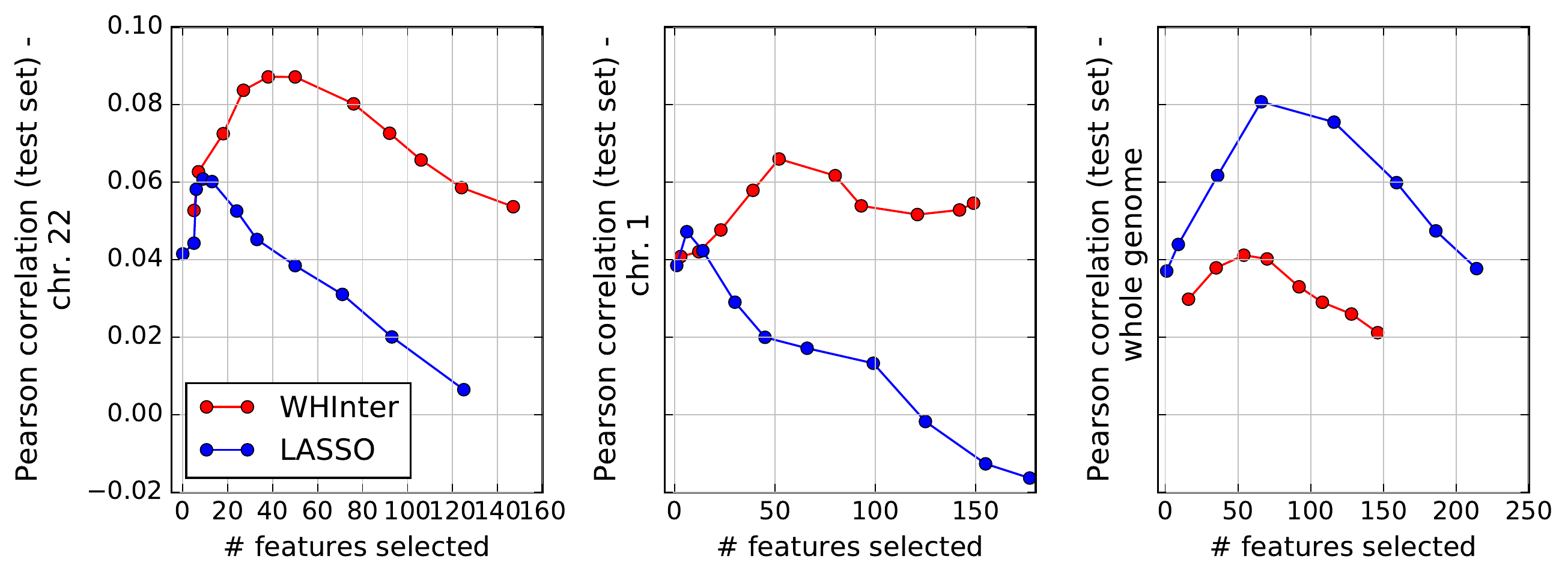}
\caption{\textbf{Predictive performance on the test set}. The $y$-axis reports the pearson correlation between the true and predicted response. The $x$-axis reports the number of selected features for the sequence of regularisation parameters tested.}
\label{fig:WHInter_realdata}
\end{figure*}

\section{Discussion}
We presented WHInter, a working set algorithm designed to solve large scale LASSO problems with interaction terms. WHInter implements a new branch pruning bound to efficiently delineate the working set among the many possible interaction variables, and a variant of MIPS solver that provides a further speed up. We showed that WHInter is up to two orders of magnitudes faster than competing approaches.
While we presented WHInter for binary data, it could also be used for data rescaled in $[0, 1]$, provided that an appropriate solver is picked for the MIPS problems. As for future work, one could exploit the recent works on approximate MIPS \citep{Shrivastava2014ALSH, Teflioudi2016LEMP} to obtain an additional speed up for the computationally intensive updates, and possibly rely on recent post selection-inference \citep{Suzuma2017SelectiveInference} frameworks to characterise the approximate solution obtained.

\section*{Acknowledgements}
We thank Nino Shervashidze for helpful discussions.

\beginsupplement
\section*{Annexes}

\subsection{Proof of Lemma~\ref{lem:bound}}\label{subsec:prooflemma}

\setcounter{Lemma}{0}
\setcounter{section}{3}
\begin{Lemma}
For any $\Xb\in\cbr{0,1}^{n\times p}$, $\vb\in\RR_+^n$, $\phib_1, \phib_2 \in \RR^n$, $j\in\bbr{p}$, $\Ical\subset\bbr{p}$ and $\alpha\in\RR$, the following holds:
\begin{equation*}\label{eq:lem1}
\underset{k \in \Ical}{\max} \abs{ \phib_2^\top (\vb \odot \Xb_k)} \leq \abs{\alpha} \underset{k \in \Ical}{\max} \abs{ \phib_1^\top (\vb \odot \Xb_k)} + \zeta(\phib_2 - \alpha \phib_1, \vb)\,,
\end{equation*}
where
$$
\forall (\ub, \vb) \in \RR^n \times \RR_+^n\,,\quad \zeta(\ub, \vb) =  \max\br{ \displaystyle{\sum_{i : \ub_i > 0} \ub_i \vb_i }, \displaystyle{-\sum_{i : \ub_i < 0} \ub_i \vb_i} } \,.
$$
\end{Lemma}

\begin{proof}
With the notations of Lemma~\ref{lem:bound} , we have:
\begin{align*}
\underset{k \in \Ical}{\max} \abs{ \phib_2^\top (\vb \odot \Xb_k)}
&\leq \underset{k \in \Ical}{\max} \abs{ \alpha \phib_1^\top (\vb \odot \Xb_k) +  (\phib_2 - \alpha\phib_1)^\top (\vb \odot \Xb_k)} \\
&\leq \abs{\alpha}\underset{k \in \Ical}{\max} \abs{ \phib_1^\top (\vb \odot \Xb_k)} + \underset{k \in \Ical}{\max} \abs{(\phib_2 - \alpha\phib_1)^\top (\vb \odot \Xb_k)} \\
&\leq \abs{\alpha}\underset{k \in \Ical}{\max} \abs{ \phib_1^\top (\vb \odot \Xb_k)} + \underset{\Xb \in \cbr{0,1}^n}{\max} \abs{(\phib_2 - \alpha\phib_1)^\top (\vb \odot \Xb)} \\
&= \abs{\alpha}\underset{k \in \Ical}{\max} \abs{ \phib_1^\top (\vb \odot \Xb_k)} +  \zeta(\phib_2 - \alpha \phib_1, \vb) \,.
\end{align*}
\end{proof}

\subsection{Computing $\eta_{min}$}\label{subsec:etamin}
\label{sec:alphamin}
In this section we characterise the existence and an algorithm to compute, for any fixed $\br{\vb,\thetab,\thetab',m} \in \RR_+^n \times \RR^n\times \RR^n \times \RR$:
\begin{equation}\label{eq:minalpha}
\eta_{min}\br{\vb,\thetab,\thetab',m} := \min_{\alpha\in\RR} \eta_{\alpha}\br{\vb,\thetab,\thetab',m} \,,
\end{equation}
where $\eta_\alpha$ is defined in Section~\ref{subsec:BB}. For that purpose, let us introduce for any $\alpha\in\RR$ the functions:
$$
\begin{cases}
\gamma_p(\alpha) &= \displaystyle{\sum_{i : \phib'_i - \alpha \thetab_i> 0} \vb_{i} \br{\phib'_i - \alpha \thetab_{i}}} \,,\\
\gamma_m(\alpha) &= \displaystyle{\sum_{i : \phib'_i - \alpha \thetab_i< 0} \vb_{i} \br{\phib'_i - \alpha \thetab_{i}}} \,,
\end{cases}
$$
such that:
\begin{equation}\label{eq:etagamma}
\eta_{\alpha}\br{\vb,\thetab,\thetab',m} = \left|\alpha\right| m + \max \br{\gamma_p(\alpha), -\gamma_m(\alpha)}\,.
\end{equation}
Let us first characterise the existence and properties of the solution to the minimisation problem (\ref{eq:minalpha}).

\begin{TheoremS}
\label{th:alpha}
For any $\br{\vb,\thetab,\thetab',m} \in \RR_+^n \times \RR^n\times \RR^n \times \RR$, the function
$$
\alpha \in \RR \rightarrow \eta_{\alpha}\br{\vb,\thetab,\thetab',m}
$$
is continuous, piecewise affine, convex and nonnegative. It reaches at least a minimum at a value $\alpha^* \in \Bcal$ where
$$
\Bcal = \cbr{0} \cup \cbr{\frac{\thetab'_i}{\thetab_i} \,:\, i\in\supp(\thetab)\cap\supp(\vb) } \cup \cbr{\alpha \in \RR: \gamma_p(\alpha) = \gamma_m(\alpha)}\,.
$$
\end{TheoremS}
\begin{proof}
For any $i\in\bbr{n}$, let
$$
\forall \alpha\in\RR\,,\quad \phi_i(\alpha) = \vb_i \max\br{0 , \thetab'_i - \alpha \thetab_i}\,.
$$
Since $\vb_i\geq 0$, $\phi_i(\alpha) = \vb_i \max\br{0 , \thetab'_i - \alpha \thetab_i }$ is continuous, piecewise affine, convex and nonnegative. It has a single breakpoint at $\alpha_i = \thetab'_i / \thetab_i$ if $\thetab_i \neq 0$ and $\vb_i >0$, and is constant otherwise. Since $\gamma_p(\alpha) = \sum_{i=1}^n \phi_i(\alpha)$, $\gamma_p$ is also continuous, piecewise affine, convex and nonnegative with breakpoints in $\cbr{\thetab'_i / \thetab_i\,:\,i \in \supp(\thetab) \cup \supp(\vb)}$. Taking $\psi_i(\alpha) = \vb_i \max\br{0 , \alpha \thetab_i - \thetab'_i }$ shows similarly that $-\gamma_m(\alpha) = \sum_{i=1}^n \psi_i(\alpha)$ has the same properties. Consequently, the function $\alpha \mapsto  \max \br{\gamma_p(\alpha), -\gamma_m(\alpha)}$ is also continuous, piecewise affine, convex and nonnegative, with possible breakpoints in
$$
\cbr{\thetab'_i / \thetab_i\,:\,i \in \supp(\thetab)\cup \supp(\vb)} \cup \cbr{\alpha \in \RR: \gamma_p(\alpha) = \gamma_m(\alpha)}\,.
$$
Since $\alpha\rightarrow\abs{\alpha}$ is also continuous, piecewise affine, convex and nonnegative, and has a breakpoint for $\alpha=0$, Theorem~\ref{th:alpha} follows by observing that a continuous, piecewise affine, convex and nonnegative function necessarily reaches a minimum at one of its breakpoints.
\end{proof}

Let $S=\abs{\supp(\thetab)\cap\supp(\vb)}$. Theorem~\ref{th:alpha} shows that it suffices to compute the values of $\eta_\alpha$ on at most $S+2$ values for $\alpha$ to find the global minimum. However, a naive computation of $\eta_\alpha$ using (\ref{eq:etagamma}) takes $O(\abs{\supp(\vb)})$ for each $\alpha$, hence a total complexity $O(S\times \abs{\supp(\vb)})$ to find the minimum of $\eta_\alpha$.

This can be improved to $O(\abs{\supp(\vb)} + S \ln S)$ by first sorting the $S+1$ breakpoints $b_i = \thetab'_i/\thetab_i$ for $i\in\supp(\thetab)\cap\supp(\vb)$ and  $b_{S+1}=0$ in increasing order:
$$
b_{\pi(1)} \leq b_{\pi(2)} \leq \ldots \leq b_{\pi(S+1)}\,,
$$
which takes $O(S \ln S)$ time. Adding by convention $b_{\pi(0)} = -\infty$ we observe that on each interval $(b_{k-1},b_k]$ the functions $\gamma_p$ and $\gamma_m$ are affine, of the form:
$$
\forall \alpha \in (b_{k-1},b_k]\,,\quad 
\begin{cases}
 \gamma_p(\alpha) &= s^k_p - \alpha t^k_p\,,\\
- \gamma_m(\alpha) &= s^k_m - \alpha t^k_m\,.\\
\end{cases}
$$
From the properties of $\gamma_p(\alpha) = \sum_{i=1}^n \phi_i(\alpha)$ and $-\gamma_m(\alpha) = \sum_{i=1}^n \psi_i(\alpha)$, we get the coefficients for $k=1$, i.e., for the interval $(-\infty,b_{\pi(1)}]$ in $O(\abs{\supp(\vb)})$ as follows:
\begin{equation}\label{eq:init}
\begin{cases}
s_p^1 &= \sum_{i\in\supp(\vb)\,:\,\thetab_i>0} \vb_i \thetab'_i + \sum_{i\in\supp(\vb)\,:\,\thetab_i=0} \vb_i \max(0,\thetab'_i)\,,\\
t_p^1 &= \sum_{i\in\supp(\vb)\,:\,\thetab_i>0} \vb_i \thetab_i\,,\\
s_m^1 &= - \sum_{i\in\supp(\vb)\,:\,\thetab_i<0} \vb_i \thetab'_i + \sum_{i\in\supp(\vb)\,:\,\thetab_i=0} \vb_i \max(0,- \thetab'_i)\,,\\
t_m^1 &= \sum_{i\in\supp(\vb)\,:\,\thetab_i<0} \vb_i \thetab_i\,.\\
\end{cases}
\end{equation}
This allows in particular to compute $\gamma_p(b_{\pi(1)})$, $\gamma_m(b_{\pi(1)})$, and therefore $\eta_{b_{\pi(1)}}$ from (\ref{eq:etagamma}). We can then iteratively compute the coefficients for $k+1$ from the coefficients for $k$ in $O(1)$ only, by observing that between the intervals $(b_{k-1},b_k]$ and $(b_{k},b_{k+1}]$, the only change in slope and intercept of $\gamma_p$ is due to the function $\phi_{\pi^{-1}(k)}$, when $\pi^{-1}(k) \neq S+1$. Let $i=\pi^{-1}(k)$. When $\thetab_i>0$, the slope of $\phi_i$ increases by $\vb_{i} \thetab_{i}$ and its intercept decreases by $\vb_i \thetab'_i$ at $b_i$. When $\thetab_i>0$, its slope increases by $-\vb_{i} \thetab_{i}$ and its intercept increases by $\vb_i \thetab'_i$. This translates into the following recursive formula for the coefficients of $\gamma_p$:
$$
s_p^{k+1} =
\begin{cases}
s_p^k - \vb_i \thetab'_i &\text{if }\thetab_i >0\,,\\
s_p^k + \vb_i \thetab'_i &\text{if }\thetab_i < 0\,,
\end{cases}
$$
and
$$
t_p^{k+1} = t_p^k - \vb_i \abs{\thetab_i}\,.
$$
A similar analysis on $\gamma_m$ leads to the following recursion:
$$
s_m^{k+1} =
\begin{cases}
s_m^k - \vb_i \thetab'_i &\text{if }\thetab_i >0\,,\\
s_m^k + \vb_i \thetab'_i &\text{if }\thetab_i < 0\,,
\end{cases}
$$
and
$$
t_m^{k+1} = t_m^k - \vb_i \abs{\thetab_i}\,.
$$
We can thus iteratively compute the coefficients on each interval, and thus the values of $\eta_\alpha$ on each breakpoint, with complexity $O(1)$ per breakpoint. Since $\alpha\mapsto\eta_\alpha$ is convex, we stop at the first $k$ such that $\eta_{b_{\pi(k+1)}} \geq \eta_{b_{\pi(k)}} $. From the equations of $\gamma_p$ and $\gamma_m$ on $(b_{\pi(k)} , b_{\pi(k+1)}]$ we can additionally check if there is a crossing point $\bar{\alpha}\in(b_{\pi(k)} , b_{\pi(k+1)}]$ such that $\gamma_p(\bar{\alpha}) = \gamma_m(\bar{\alpha})$, in which case we also compute $\eta_{\bar{\alpha}}$. The global minimum of $\alpha\mapsto\eta_\alpha$ is then $\min( \eta_{b_{\pi(k)}} ,  \eta_{\bar{\alpha}})$.

The overall algorithm is detailed in Algorithm~\ref{alg:find_alpha}.

\begin{algorithm}
\caption{Minimise $\eta$ in $\alpha$}
\label{alg:find_alpha}
\begin{algorithmic}[1]
\REQUIRE $\br{\vb,\thetab,\thetab',m} \in \RR_+^n \times \RR^n\times \RR^n \times \RR$.
\ENSURE $\eta_{min}\br{\vb,\thetab,\thetab',m}$
\STATE  $S \leftarrow$ indices in $\supp(\vb) \cap \supp(\thetab)$
\STATE $N \leftarrow$ length($S$)
\STATE $v\leftarrow \sqb{0, \frac{\thetab'_{S[1]}}{\thetab_{S[1]}}, \dots, \frac{\thetab'_{S[N]}}{\thetab_{S[N]}}}$
\STATE $\texttt{ind} \leftarrow \sqb{\texttt{none}, S[1], \dots, S[N]}$
\STATE $\texttt{rank} \leftarrow$ sort($v$) (in increasing order)
\STATE $v \leftarrow v[\texttt{rank}]; \; \texttt{ind} \leftarrow \texttt{ind} [\texttt{rank}]$
\STATE Initialise $s_p$, $s_m$, $t_p$, $t_m$ via (\ref{eq:init})
\STATE $\texttt{min} \leftarrow + \infty$
\FOR{$i$ in $1 \dots N + 1$}
\STATE $\texttt{newmin} \leftarrow \left|v[i]\right|m + \max \br{s_p - v[i]t_p, \; s_m - v[i] t_m}$
\IF{$\texttt{newmin} < \texttt{min}$}
\STATE $\texttt{min} \leftarrow \texttt{newmin}$
\IF{$\texttt{ind}[i] \ne \texttt{none}$}
\STATE  $t_p \leftarrow t_p - \vb_{\texttt{ind}[i]} \abs{\thetab_{\texttt{ind}[i]}} $
\STATE  $t_m \leftarrow t_m - \vb_{\texttt{ind}[i]} \abs{\thetab_{\texttt{ind}[i]}} $
\IF{ $\thetab_{\texttt{ind}[i]} > 0$}
\STATE  $s_p \leftarrow s_p - \vb_{\texttt{ind}[i]}\thetab'_{\texttt{ind}[i]} $
\STATE $s_m \leftarrow s_m - \vb_{\texttt{ind}[i]}\thetab'_{\texttt{ind}[i]} $
\ELSE
\STATE $s_p \leftarrow s_p + \vb_{\texttt{ind}[i]}\thetab'_{\texttt{ind}[i]} $
\STATE  $s_m \leftarrow s_m + \vb_{\texttt{ind}[i]}\thetab'_{\texttt{ind}[i]} $
\ENDIF
\ENDIF
\ELSE
\STATE Check if there exists $\bar{\alpha}\in[v[i-1],v[i]]$ s.t. $\gamma_p(\bar{\alpha}) = \gamma_m(\bar{\alpha})$
\STATE Return min($\texttt{newmin}, \, \eta(\alpha^{intersection}))$
\ENDIF
\ENDFOR

\end{algorithmic}
\end{algorithm}

\subsection{Alternative solver for working set updates}\label{subsec:MIPS1}
In this section, we present an alternative solver to the inverted list approach (algorithm~\ref{alg:MIPS2} in section~\ref{subsec:MIPS}), which we call $MIPS1$, to compute the working set updates (\ref{eq:updates}). It relies on a pruning technique and does not require storing extra indices for the data. The main idea of this alternative approach is to compute inner products on a progressively growing subset of dimensions, and to maintain an upper-bound on the maximum attainable score on the remaining dimensions. This allows to discard a probe as soon as its maximum attainable score drops below the maximum score achieved so far without computing the inner product in its entirety. Algorithm \ref{alg:MIPS1} presents the procedure in details. It takes as input $\Qcal$ which contains the indices that define the queries of interest and outputs the updated working set $\Wcal$ and $\mb^{ref}$. For each query, we start by precomputing the partial inner product bounds $\rb^+ \in \RR^n$ and $\rb^- \in \RR^n$, where $\rb_i^+$ and $\rb_i^-$ are respectively the maximum and minimum attainable inner products between the query and any probe in the database on the dimensions from $i + 1$ to $n$. Formally, $\rb^+$ and $\rb^-$ are defined for a given query $j$ by:
\begin{align}
\label{eq:rp}
\forall i\in \bbr{n},  r_i^+ &= \sum_{m > i;\ \phib_m > 0} \Xb_{mj} \phib_m\\
\label{eq:rm}
\forall i\in \bbr{n},  r_i^- &= \sum_{m > i;\ \phib_m < 0} \Xb_{mj} \phib_m
\end{align}
and provide an upper bound on inner products with the query $\Xb_j \odot \phib$ as follows:
\begin{align*}
\forall k \in \bbr{p}, \quad \left(\Xb_j \odot \phib\right)^\top \Xb_k & = \sum_{m \leq i} \Xb_{mj} \phib_m \Xb_{mk} + \sum_{m > i} \Xb_{mj} \phib_m \Xb_{mk}\\
& \leq \sum_{m \leq i} \Xb_{mj} \phib_m \Xb_{mk} + \sum_{m > i;\ \phib_m > 0} \Xb_{mj} \phib_m\\
& = \sum_{m \leq i} \Xb_{mj} \phib_m \Xb_{mk} + r_i^+
\end{align*}
The bound involving $\rb^-$ can be obtained analogously. These bounds simply assume there is a probe vector which has ones in front of every positive entry of the query and none in front of its negative entries, or the reverse. Once these bounds have been precomputed, the inner product between the query and a probe is computed up to a certain dimension, and every $n_c \in \mathbb{N}$ dimensions we check whether there is a possibility that the inner product being computed becomes larger than the current maximum, or larger than $\lambda$. If it is impossible, then the probe can be safely discarded and the algorithm proceeds with the next probe. If not, the inner product is computed on $n_c$ more dimensions and a new check is performed. For all our simulations and real data experiments, we set $n_c$ to a default of 20. If a probe cannot be discarded then the algorithm updates when appropriate the active set $\Wcal$ and/or the current maximum absolute inner product obtained $\mb_j^{ref}$. For the pruning to be effective, we reorder the dimensions $1 \dots n$ so that queries are sorted in decreasing order in absolute value. As a consequence, the partial inner product bounds $\rb_i^+$ and $\rb_i^-$ are computed with the $n-i$ smallest entries in absolute value of the queries which makes them tighter than with any other ordering of the dimensions.\\

\begin{algorithm}
\caption{MIPS1}
\label{alg:MIPS1}
\begin{algorithmic}[1]
\REQUIRE $\Xb \in [0, 1]^{n \times p}, \,\phib \in \RR^n, \, \Qcal \subset \bbr{p}, \, \lambda \in \RR, \, \Wcal \subset \bbr{D}$
\PARAM $n_c \in \NN$
\ENSURE $\Wcal$, $\mb^{ref}$.
\STATE Reorder the dimensions $1\dots n$ such that $\thetab$ is sorted in descending order in absolute value and reorder the dimensions of $\Xb$ accordingly.
\STATE Reorder the columns of $\Xb$ in descending order of vector size.
\STATE \textbf{for} $j \in \Qcal$ \textbf{do} $\mb_j^{ref} \leftarrow 0$
\FOR{$j \in  \Qcal$}
\STATE Compute $\rb^+ \in \RR^n$ and  $\rb^- \in \RR^n$ via (\ref{eq:rp}) and (\ref{eq:rm}).
\FOR{$k \in  \bbr{p} $}
\STATE \textbf{if}  $k \in \Qcal$ \text{and} $k > j$ \text{then} \texttt{continue}
\STATE $d \leftarrow 0$ (inner product initialization); c = 0 (counter initialization);
\FOR{$i \in$  supp($\bm X_j$)}
\STATE $d \leftarrow d + \Xb_{ij} \Xb_{ik} \phib_i$
\STATE $c \leftarrow c + 1$.
\IF{ $c \texttt{ mod } n_c = 0$}
\STATE \textbf{if} $(d + \rb^+_i) < \text{min}(\mb^{ref}_j, \lambda) $ \texttt{and} $\left| (d + \rb^-_i)\right| < \text{min}(\mb_j^{ref} , \lambda)$ \textbf{then} go to next probe.
\ENDIF
\ENDFOR
\STATE \textbf{if}  $\mb^{ref}_j < d < \lambda$ \textbf{then} set $\mb^{ref}_j = d$
\STATE \textbf{if}  $d \geq \lambda$ and $\tau(k, j) \notin \Wcal$ \textbf{then} add $\tau(k, j)$ to $\Wcal$
\ENDFOR
\ENDFOR
\textbf{return} $\Wcal, \bm m^{ref}$
\end{algorithmic}
\end{algorithm}

We now compare $MIPS1$ to its naive counterpart (which we will call $Naive$ from now on) on several benchmark datasets in order to assess the speed-up obtained with the pruning. To be more specific, $Naive$ is implemented similarly to $MIPS1$ except the lines specific to pruning, i.e., lines 5, 12 and 13 in Algorithm \ref{alg:MIPS1}, are removed. The benchmark datasets we use are designed in such a way that the pruning rate achievable varies. To do this, we simulate a matrix $\Xb \in \RR^{n \times p}$, with $n=p=1000$, where the features are drawn from a Bernoulli distribution, whose parameter is itself drown from a uniform distribution $\Ucal_{[0.1, 0.5]}$. Then $\phib \in \RR^n$ is built in such a way that the cumulative sum of the vectors obtained by sorting $\phib_{|\phib \geq 0}$ and  $|\phib_{|\phib < 0}|$ follows the function $f(x) = \frac{1}{1 - e^{-\mu}} \left( 1 - e^{-\mu x}\right), x \in \cbr{0, 1}$ for a given parameter $\mu \in \RR^+$. The area under this cumulative sum, which is $\kappa(\mu) = \frac{1}{1 - e^{-\mu}} - \frac{1}{\mu} \in [0.5, 1]$, characterises the different vectors $\phib_\kappa$ obtained with different values of $\mu$. Figure~\ref{fig:cumsum} shows how the cumulative sums are modified with $\mu$. The interest of simulating different $\phib_\kappa$ is that the rate of pruning achievable increases with $\kappa$: the closer $\kappa$ is to 1, the higher the pruning rate. In the experiments presented hereafter, all $p$ features were taken as queries, i.e., $\Qcal = \bbr{p}$, and we took $\lambda=+\infty$ and $\Wcal = \emptyset$. The results are presented in Figure~\ref{fig:ratio}. The pruning rate, which we define as the average number of non-zero coordinates of the queries which were pruned out of their total number of non-zero coordinates, widely varies from 8\%  for $\kappa = 0.55$ to 84\% for $\kappa = 0.95$. Moreover, the speed-up obtained with $MIPS1$ compared to $Naive$ is almost equal to 1 minus the pruning rate. That means $MIPS1$ is twice as fast as $Naive$ when it can prune half of the total number of coordinates.\\

\begin{figure}
\begin{subfigure}[t]{0.55\linewidth}
	\caption{}\label{fig:cumsum}
	\includegraphics[scale=0.7]{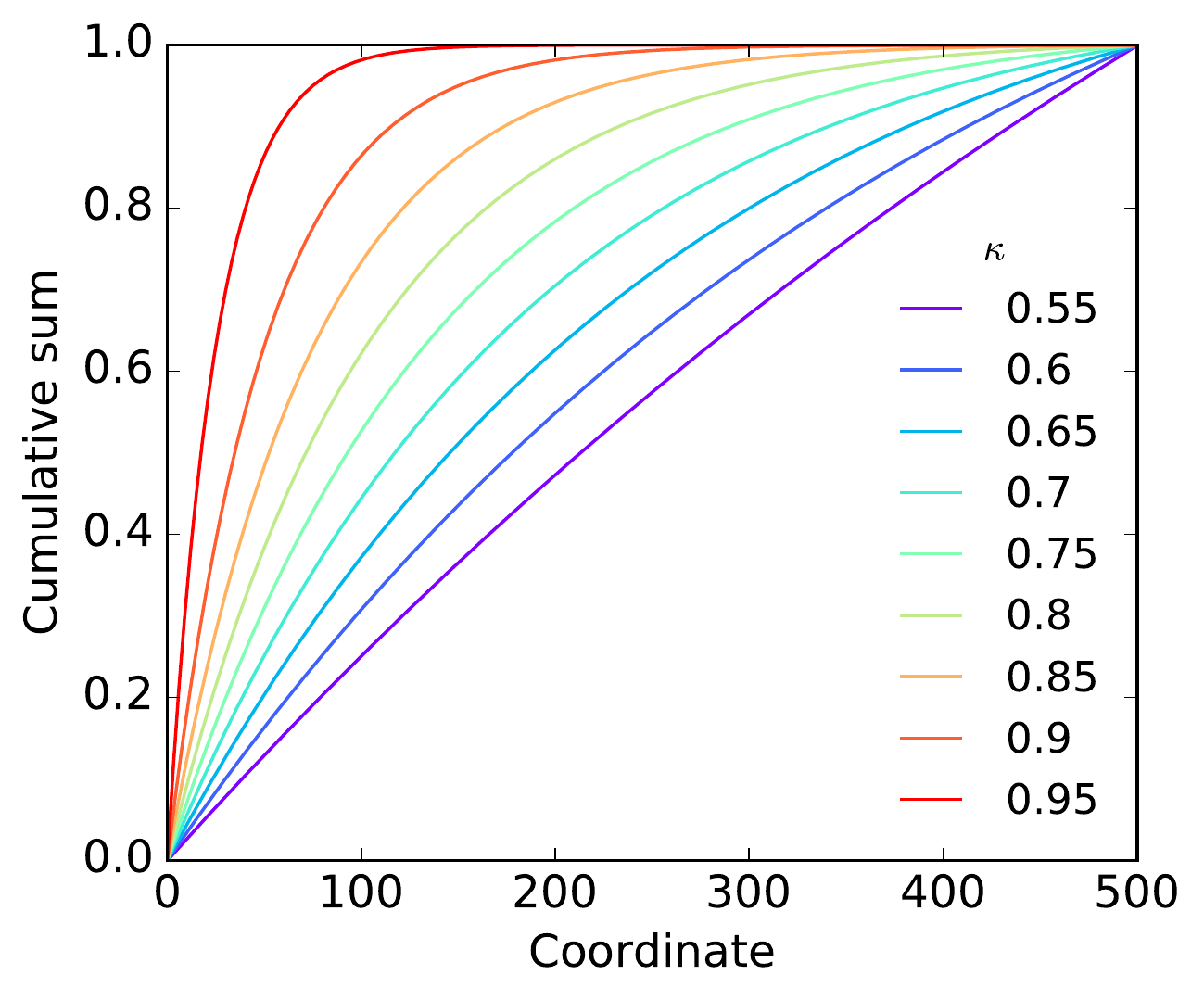}	
\end{subfigure}
\quad
\begin{subfigure}[t]{0.45\linewidth}
	\caption{}\label{fig:ratio}
	\includegraphics[scale=0.7]{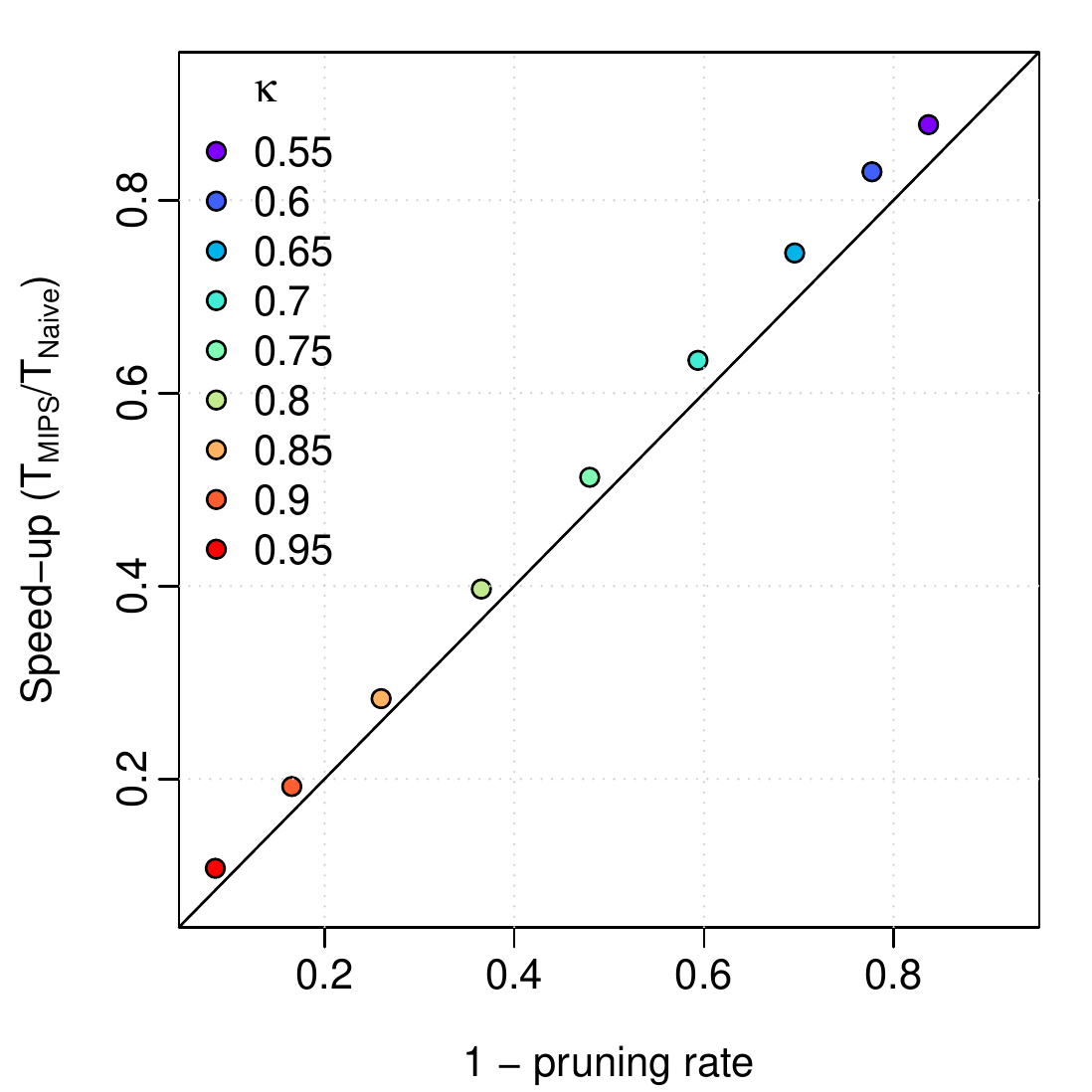}
\end{subfigure}
\caption{ (\subref{fig:cumsum}) Cumulative sum of the vector obtained by sorting the positive entries of $\phib_\kappa$ in decreasing order. (\subref{fig:ratio}) Speed-up obtained with $MIPS1$ compared to $Naive$ for different vectors $\phib_\kappa$ as a function of the pruning rate. The pruning rate is defined as the average proportion of coordinates in the queries which are pruned.}
\label{fig:test_MIPS}
\end{figure}

We now compare the performance of $Naive$, $MIPS1$ and $IL$ on the benchmark datasets (Figure \ref{fig:test_MIPS_time}). $MIPS1$ is the only method whose speed depends on $\kappa$ since it is the only method to implement pruning. It has the same performance in terms of speed as $Naive$ for the lowest pruning rate, while it is as fast as $IL$ for the highest pruning rates. For vectors $\phib$ following classical distributions such as the gaussian distribution, $\kappa \approx 0.7$ and $MIPS1$ is therefore expected to be $\times 1.6$ times faster than $Naive$ but $\times 11$ times slower than $IL$. An analysis of the complexity of $MIPS1$ and $IL$ can help to understand these results. For a given query, $MIPS1$ requires to compute inner products (although partially) with all $p$ vectors in the database. In our implementation, the vectors are encoded as sparse vectors, i.e., the vector $\Xb_j$ is represented by the list of its non-zero indices. If we assume that the number of non-zero elements in the query is $|q|$ and that the total number of non-zero elements of the vectors in $\Xb$ in $nnz$, then $MIPS1$ has a $O(p|q| + nnz)$ complexity to compute the $p$ inner products with the query. By contrast, the inverted index approach has a $O(|q|\frac{nnz}{n})$ complexity, where $\frac{nnz}{n}$ is the average length of an inverted index. As the number of non-zero elements $|q|$ in the query will typically be a fraction of the total number of samples $n$, the inverted index approach is expected to be faster than $MIPS1$ even though the pruning in $MIPS1$ can make it faster. This however may not be the case with dense data instead of sparse data.

\begin{figure}
\centering
\includegraphics[scale=0.7]{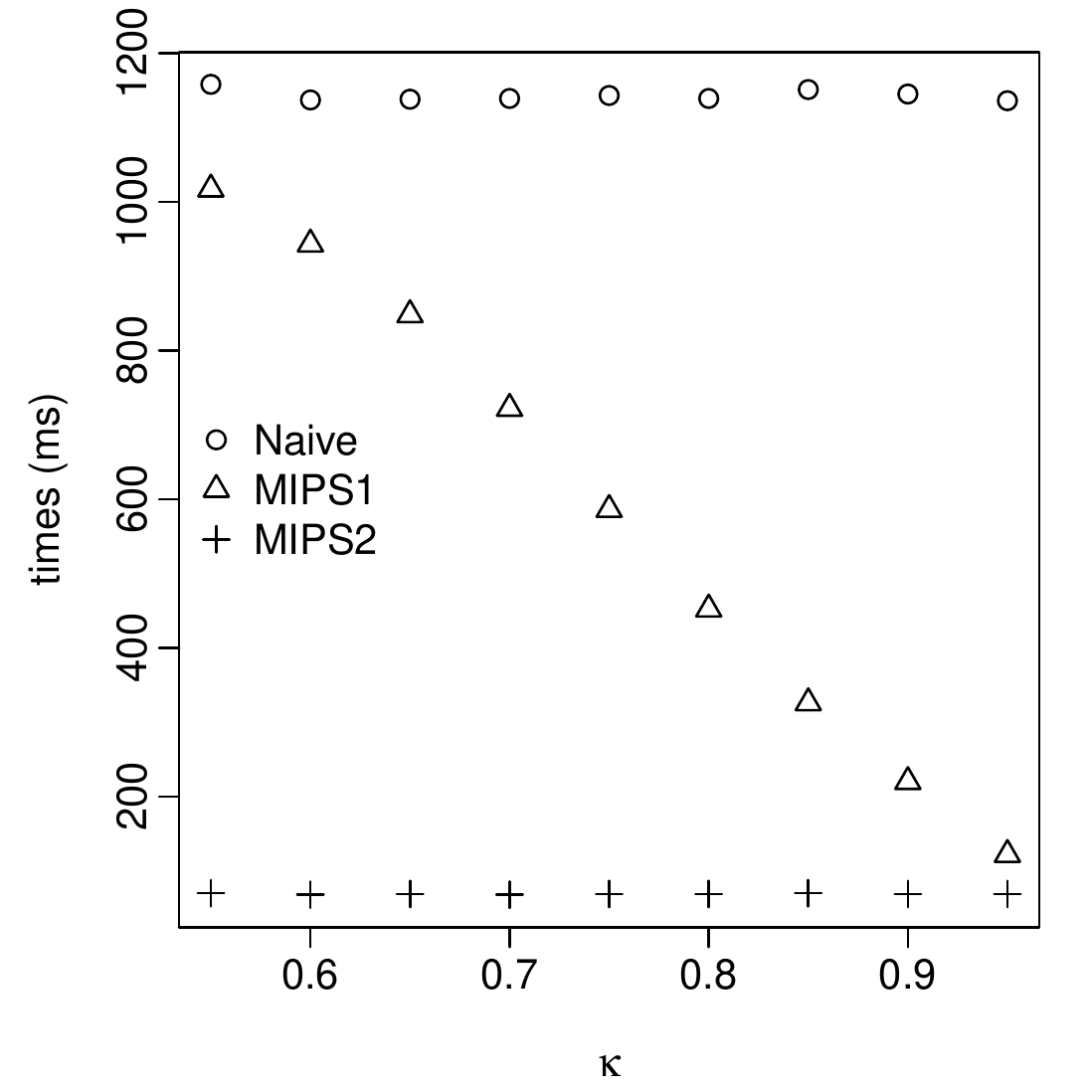}
\caption{Time (in ms) taken by $Naive$, $MIPS1$ and $MIPS2$ to solve Maximum Inner Product Search problems with responses characterised by different $\kappa$.}
\label{fig:test_MIPS_time}
\end{figure}

\subsection{SPP: depth-first vs breadth-first}\label{subsec:breadthfirst}
The Safe Pattern Pruning algorithm presented in \cite{Nakagawa2016SPP} deals with pairwise interactions but also higher-order interactions, and relies on a depth-first search scheme to explore the tree of patterns. However in our setting where we only consider pairwise interactions, we find that it is more efficient to implement a breadth-first search scheme for SPP. Indeed, the breadth-first search first identifies all the branches which can be screened. Then with this knowledge, we can restrict the number of interactions which are visited to those which only involve main effects whose corresponding branch was not screened.  Basically, if we consider a case where $p_s$ branches were screened among $p$ branches, then the total number of nodes visited will be $p + \frac{(p-p_s)(p-p_s-1)}{2}$. Figure (\ref{fig:SPP_time_comparison}) illustrates the difference in performance obtained with the original SPP and the breadth-first search version in the case of pairwise interactions. The speed up obtained with the breadth-first search version ranges from $\times 1.2$ for $n=p=1000$ to $\times 1.6$ for $n=1000, p=10000$. We therefore use the breadth-first search version of SPP as a comparison baseline in all our experiments. 

\begin{figure}
\begin{subfigure}[t]{0.5\linewidth}
	\caption{}\label{fig:SPP_n}
	\includegraphics[scale=0.65]{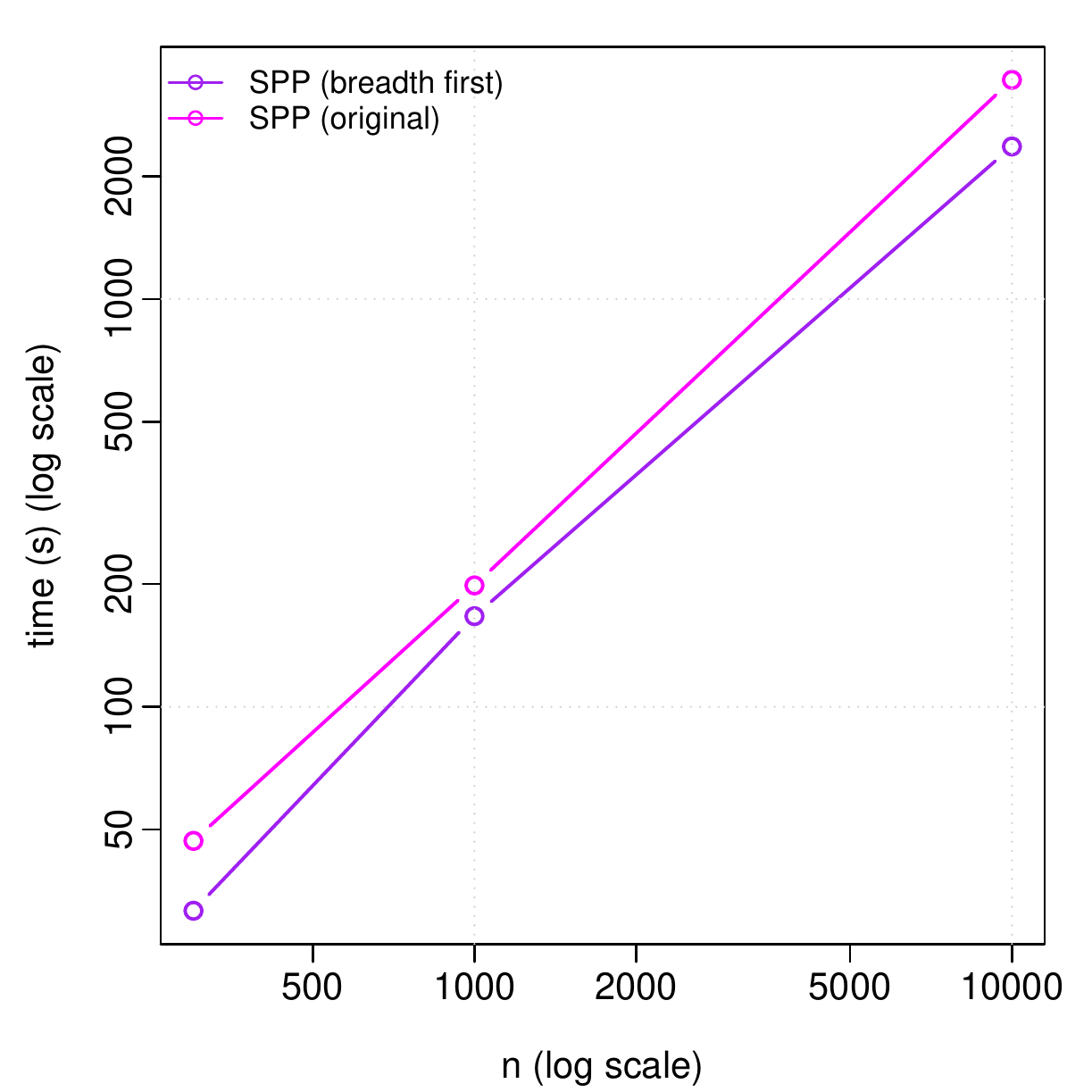}	
\end{subfigure}
\quad
\begin{subfigure}[t]{0.5\linewidth}
	\caption{}\label{fig:SPP_p}
	\includegraphics[scale=0.65]{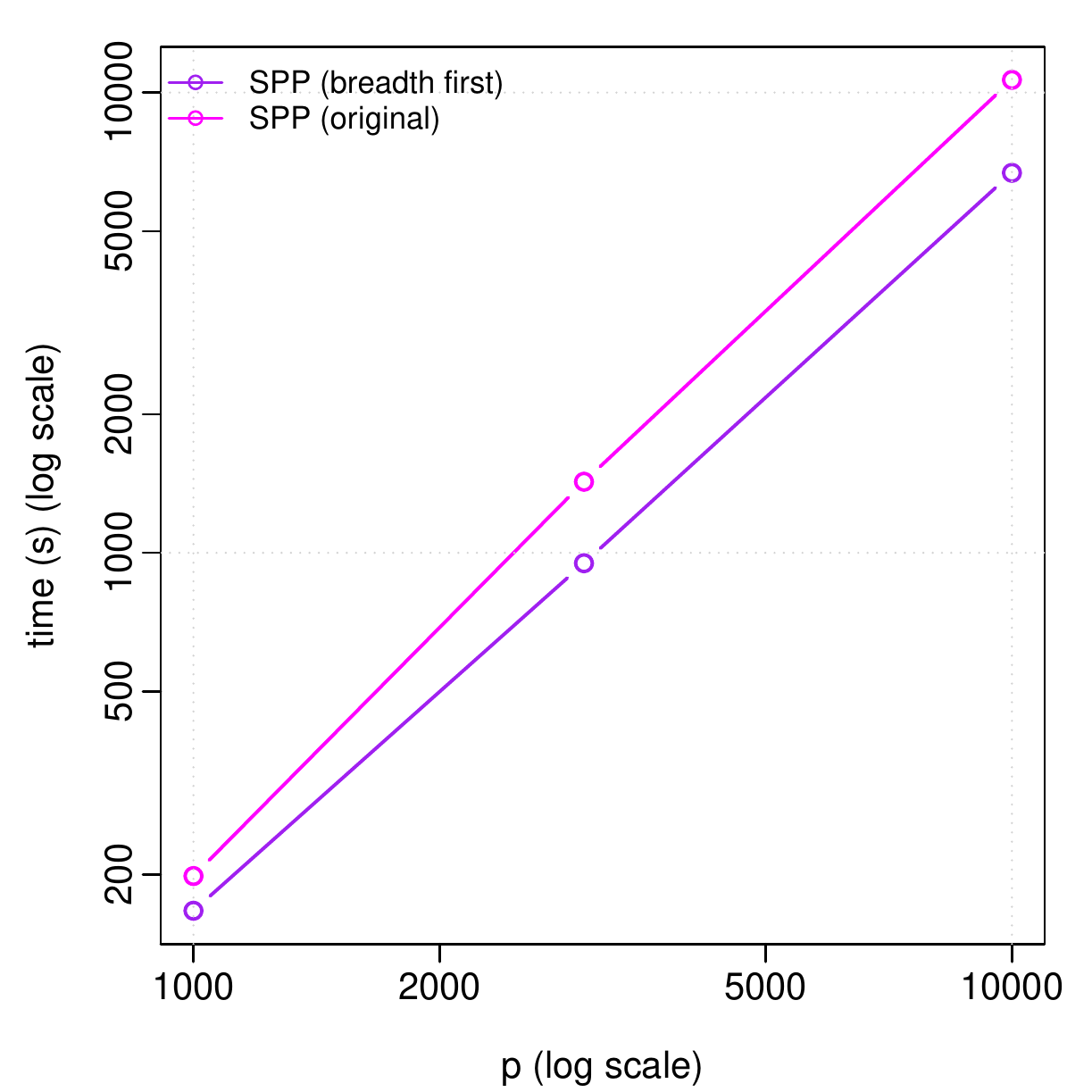}
\end{subfigure}
\caption{\textbf{Safe Pattern Pruning performance on simulated data for an entire regularisation path}. The breadth-first search SPP (which is adapted to order-2 interactions only) is in purple and the original depth-first search SPP (which is adapted to order-2 interactions and more) is in magenta. (\subref{fig:SPP_n}) Time in seconds for $p=1000$ fixed and $n$ varied. (\subref{fig:SPP_p}) Time in seconds for $n=1000$ fixed and $n$ varied.}
\label{fig:SPP_time_comparison}
\end{figure}

\bibliography{WHInter}
\end{document}